\documentclass[11pt,a4paper]{article}

\usepackage[a4paper,margin=1in]{geometry}
\usepackage[T1]{fontenc}
\usepackage{lmodern}
\usepackage{microtype}

\usepackage{amssymb,latexsym,mathtools,amsfonts,amsthm}
\usepackage{amsbsy}
\usepackage{bm}
\usepackage{mathrsfs}
\usepackage{extarrows}
\usepackage{enumerate}
\usepackage{enumitem}
\usepackage{float}
\usepackage{multirow}

\usepackage{graphicx}
\usepackage{overpic}
\usepackage{tikz}
\usetikzlibrary{arrows.meta,calc,decorations,decorations.markings}

\tikzset{
  boundarycircle/.style={
    draw=black,
    line width=1.25pt,
    dash pattern=on 6pt off 6pt
  },
  line/.style={
    draw=black,
    line width=1.6pt,
    line cap=round,
    line join=round
  },
  axis/.style={
    line,
    line width=1.25pt,
    dash pattern=on 5pt off 6pt
  },
  flow/.style={
    line,
    postaction={decorate},
    decoration={
      markings,
      mark=at position #1 with
        {\arrow{Stealth[length=5.3pt,width=5.8pt]}}
    }
  },
  flow/.default=.55,
  twoflow/.style args={#1 and #2}{
    line,
    postaction={decorate},
    decoration={
      markings,
      mark=at position #1 with
        {\arrow{Stealth[length=5.3pt,width=5.8pt]}},
      mark=at position #2 with
        {\arrow{Stealth[length=5.3pt,width=5.8pt]}}
    }
  },
  axisflow/.style={
    axis,
    postaction={decorate},
    decoration={
      markings,
      mark=at position #1 with
        {\arrow{Stealth[length=5.0pt,width=5.5pt]}}
    }
  },
  axisflow/.default=.88,
  label/.style={
    font=\small,
    inner sep=1pt
  },
  regionlabel/.style={
    font=\normalsize,
    inner sep=1pt
  },
  contourlabel/.style={
    font=\normalsize,
    inner sep=1pt
  }
}
\usepackage{xcolor}
\usepackage{cite}
\usepackage{hyperref}

\hypersetup{
    hidelinks,
    colorlinks=true,
    linkcolor=blue,
    citecolor=blue
}

\newtheorem{theorem}{Theorem}[section]
\newtheorem{lemma}[theorem]{Lemma}

\newtheorem{assume}[theorem]{Assumption}
\newtheorem{RHP}[theorem]{RH Problem}

\theoremstyle{remark}
\newtheorem{remark}[theorem]{Remark}

\numberwithin{equation}{section}
\allowdisplaybreaks
\newcommand{\R}{\mathbb{R}}
\newcommand{\C}{\mathbb{C}}

\newcommand{\ii}{\mathrm{i}}

\newcommand{\e}{\mathrm{e}}

\newcommand{\Q}{\mathcal{Q}}

\newcommand{\cC}{\mathcal{C}}

\newcommand{\cD}{\mathcal D}
\newcommand{\cR}{\mathcal R}
\newcommand{\cP}{\mathcal{P}}

\newcommand{\eps}{\epsilon}

\DeclareMathOperator{\diag}{diag}

\def\be{\begin{equation}}
\def\ee{\end{equation}}
\def\bse{\begin{subequations}}
\def\ese{\end{subequations}}
\def\bpm{\begin{pmatrix}}
\def\epm{\end{pmatrix}}
\def\bi{\begin{itemize}}
\def\ei{\end{itemize}}

\title{\bfseries
Tritronqu\'ee Painlev\'e II asymptotics for the focusing nonlinear Schr\"odinger equation with nonzero boundary conditions}
\date{}

\author{\hspace{0.6 cm}{Haibing Zhang$^{a}$, Xianguo Geng$^{a,b}$\footnote{\footnotesize
 Corresponding author. {\sl Email address}: xggeng@zzu.edu.cn}, Kedong Wang$^{a}$}\\
\leftline{\hspace{0.6 cm}{\small{\sl $^{a}$ School of Mathematics and Statistics, Zhengzhou University, 100 Kexue Road, Zhengzhou, }}}\\
\leftline{\hspace{0.6 cm}{\small{\sl \quad Henan 450001, People's Republic of China}}}\\
\leftline{\hspace{0.6 cm}{\small{\sl $^{b}$ School of Mathematics and Statistics, North China University of Water Resources }}}\\
\leftline{\hspace{0.6 cm}{\small{\sl \quad and Electric Power, Zhengzhou, Henan 450011, People's Republic of China}}}\\}

\begin{document}
\maketitle
\begin{abstract}
We study the long-time asymptotics of the focusing nonlinear Schr\"odinger
equation with nonzero boundary conditions in the transition regions between
the plane-wave and modulated elliptic-wave regimes.  Biondini and
Mantzavinos showed that, away from the transition curves
\(x=\pm 4\sqrt{2}\,q_o t\), the \((x,t)\)-half-plane decomposes, to leading
order, into two plane-wave regions and a central region described by slowly
modulated elliptic oscillations.  However, their asymptotic formulae are not
uniform near the boundaries separating these regions.  The purpose of this
paper is to resolve this missing boundary layer.

Using a double-scaling nonlinear steepest descent analysis of the associated
Riemann--Hilbert problem, we show that the leading term in each transition
region is still a plane wave, while the first nontrivial correction is of order
\(t^{-1/3}\).  The coefficient of this correction is expressed in terms of a
distinguished tritronqu\'ee solution of an inhomogeneous Painlev\'e-II equation.
This Painlev\'e-II tritronqu\'ee structure is also known to appear in the
asymptotic analysis of rogue waves of infinite order.
\end{abstract}

\vspace{0.2cm}

\vspace{0.2cm}

\tableofcontents

\section{Introduction}

The one-dimensional focusing nonlinear Schr\"odinger (NLS) equation
\be \label{E:wbhFNLS}
\ii q_t+q_{xx}+2 |q|^2q=0,
\qquad x\in\R,\quad t\ge 0,
\ee
is one of the central model equations in nonlinear wave theory and integrable
systems.  It has been studied extensively for more than half a century; see,
for example, the monographs~\cite{APT2004,AH1981,EG2003,CP1999} and the
references therein.  In particular, the focusing NLS equation plays a
distinguished role in the study of modulational instability~\cite{ZO2009,BF2015,BLMT2018}.
Modulational instability, also known as the Benjamin--Feir instability in the
context of water waves~\cite{B1967-1,B1967-2}, refers to the instability of a
constant background under long-wavelength perturbations.  At the linearized
level, sufficiently low-wavenumber modes grow exponentially in time.  This
linear description, however, only captures the early stage of the instability.
Once the perturbation becomes comparable with the background, nonlinear effects become essential.  The subsequent evolution is usually
referred to as the nonlinear stage of modulational instability.

A standard model for studying modulational instability on the infinite line is the focusing cubic NLS equation with symmetric nonzero boundary conditions (NZBCs): \be\label{E:wbhbjtj} 
\lim_{x\to\pm\infty} q(x,t)=q_\pm e^{2\ii q_o^2t}, \qquad t\ge 0, 
\ee 
where \(q_\pm\) are constants satisfying \(q_\pm=q_o e^{\ii\theta_\pm}\), with \(q_o>0\) and \(\theta_\pm\in[0,2\pi)\). After a standard gauge transformation, \eqref{E:wbhFNLS}--\eqref{E:wbhbjtj} is equivalent to the following initial-value problem~(IVP) with time-independent NZBCs: 
\begin{subequations}\label{ds-fnls-ivp} \begin{align} &\ii q_t+q_{xx}+2\left(|q|^2-q_o^2\right)q=0, \qquad && x\in\R,\quad t>0, \label{ds-fnls-ivp-fnls-eq} \\ &q(x,0)=f(x), && x\in\R, \\ &\lim_{x\to\pm\infty}q(x,t)=q_\pm, && t\ge 0. \label{ds-fnls-ivp-bc} 
\end{align} \end{subequations} 
This is the main problem studied in the present paper.
Well-posedness results for the IVP~\eqref{ds-fnls-ivp} are available by
harmonic analysis methods.  For instance, Mu\~noz~\cite{m2017} proved local
well-posedness for perturbations of the plane-wave background in Sobolev spaces
\(H^s\), \(s>1/2\).  Related stability and instability results for NLS
breathers on nonzero backgrounds can be found in~\cite{AFM2019,AFM2020}.

The complete integrability of the NLS equation~\cite{zs1971} allows one to analyze
IVP~\eqref{ds-fnls-ivp} by means of the inverse scattering transform (IST),
which can be regarded as a nonlinear analogue of the Fourier transform. The corresponding IST for~\eqref{ds-fnls-ivp} was developed by Biondini
and Kova\v{c}i\v{c}~\cite{BK2014}.  A central outcome of their work was to
associate the solution of IVP~\eqref{ds-fnls-ivp} with a \(2\times2\) matrix
Riemann--Hilbert (RH) problem, from which the potential can be recovered by an appropriate reconstruction formula.
This RH formulation provides the starting point for the rigorous analysis of
the long-time behavior of solutions to IVP~\eqref{ds-fnls-ivp}. The main
asymptotic tool used for this purpose is the Deift--Zhou nonlinear steepest
descent method~\cite{pz93}, which is a fundamental framework for the
large-parameter analysis of oscillatory RH problems.  
This method has been widely used to study long-time asymptotics in integrable
systems, and it has also played an important role in orthogonal polynomial
theory and random matrix theory; see, for example,
\cite{LLY2026-JDE,DLM2020,DP2024,boutet2011,CLPa2020,BJM2018,PZ2002,Monvel5,BIS2010,WF2023,CuJe2016,
CL2024main,WZZ2026,HWZ2026,YTL2025,FLYZ2026,CLW2023,CJ2024,DKMVZ1999,Deift1999}.
In the present setting, the combination of the IST with the Deift--Zhou method
provides an effective and systematic approach to  the nonlinear stage of modulational instability for generic perturbations
of the constant background.

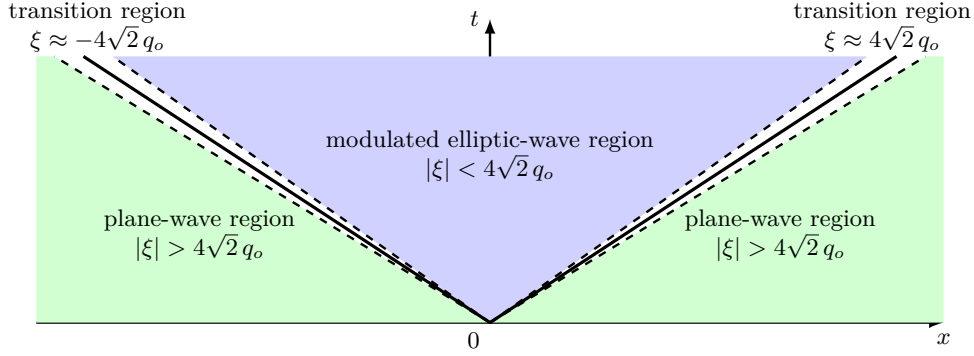
\begin{figure}
\centering
\begin{tikzpicture}[x=0.30cm,y=1.10cm,>=latex,font=\footnotesize]

\def\tmax{3.2}          
\def\xmax{20.0}         
\def\xiMain{5.6}        
\def\xiIn{5.2}          
\def\xiOut{6.0}         

\pgfmathsetmacro{\xmain}{\xiMain*\tmax}
\pgfmathsetmacro{\xin}{\xiIn*\tmax}
\pgfmathsetmacro{\xout}{\xiOut*\tmax}

\draw[->, line width=0.9pt] (-\xmax,0) -- (\xmax,0) node[below] {\(x\)};
\draw[->, line width=0.9pt] (0,0) -- (0,3.65) node[left] {\(t\)};


\fill[blue!18]
  (0,0) -- (\xin,\tmax) -- (-\xin,\tmax) -- cycle;

\fill[green!18]
  (0,0) -- (-\xmax,0) -- (-\xmax,\tmax) -- (-\xout,\tmax) -- cycle;

\fill[green!18]
  (0,0) -- (\xmax,0) -- (\xmax,\tmax) -- (\xout,\tmax) -- cycle;


\draw[dashed, line width=0.9pt] (0,0) -- (\xin,\tmax);
\draw[dashed, line width=0.9pt] (0,0) -- (\xout,\tmax);

\draw[dashed, line width=0.9pt] (0,0) -- (-\xin,\tmax);
\draw[dashed, line width=0.9pt] (0,0) -- (-\xout,\tmax);

\draw[line width=1.2pt] (0,0) -- (\xmain,\tmax);
\draw[line width=1.2pt] (0,0) -- (-\xmain,\tmax);


\node[align=center] at (0,2.0)
{
modulated elliptic-wave region\\
\(|\xi|<4\sqrt{2}\,q_o\)
};

\node[align=center] at (12.8,1.05)
{
plane-wave region\\
\(|\xi|>4\sqrt{2}\,q_o\)
};

\node[align=center] at (-12.8,1.05)
{
plane-wave region\\
\(|\xi|>4\sqrt{2}\,q_o\)
};

\node[align=center] at (17.3,3.55)
{
transition region\\
\(\xi\approx 4\sqrt{2}\,q_o\)
};

\node[align=center] at (-17.3,3.55)
{
transition region\\
\(\xi\approx -4\sqrt{2}\,q_o\)
};

\node[below left] at (0,0) {\(0\)};

\end{tikzpicture}
\caption{The space-time regions for the focusing NLS equation. The two solid rays correspond to $\xi=\pm 4\sqrt{2}\,q_o$. The central wedge represents the modulated elliptic-wave
region, the outer domains correspond to the plane-wave regions, and the narrow
bands indicated by dashed lines denote the transition regions.}
\label{fig:regions}
\end{figure}

In Ref.~\cite{BM2017}, Biondini and Mantzavinos adapted the Deift--Zhou
nonlinear steepest descent method, together with a suitable \(g\)-function
mechanism, to the long-time asymptotic analysis of
IVP~\eqref{ds-fnls-ivp}.  Their work provided the first rigorous description
of the asymptotic stage of modulational instability for generic localized
perturbations of a nonzero constant background, under the assumptions that no
discrete spectrum is present and that the initial data converge sufficiently
rapidly to the boundary values.  More precisely, they assumed an exponential
decay condition of the form
\be\label{ds-fnls-ic-space}
\e^{\pm q_o x}\bigl(f-q_\pm\bigr)\in L^1(\R_\pm),
\ee
which is also imposed throughout the present work. In terms of the self-similar variable
\[
\xi=\frac{x}{t},
\]
they showed that, to leading order, the \((x,t)\)-half-plane decomposes into
two qualitatively different types of regions.  In the two far-field regions
\(|x|>4\sqrt{2}\,q_o t\), the solution is asymptotically a plane wave with the
same amplitude as the corresponding boundary value, up to a phase shift
determined by the reflection coefficient.  Inside the central cone
\(|x|<4\sqrt{2}\,q_o t\), the leading-order behavior is no longer a
constant-amplitude wave; instead, it is described by a slowly modulated
elliptic wave.  The modulation parameters are determined by a system of
Whitham-type equations and are independent of the fine details of the initial
perturbation, while the initial data enter only through phase and
position-type shifts.  In this sense, their result shows that\textit{ the nonlinear stage of modulational instability is universal}. 
This asymptotic picture was later extended by Biondini, Li, and
Mantzavinos~\cite{BLM2021} to the case in which the scattering data contain a
conjugate pair of discrete eigenvalues.  Their analysis describes the
interaction between a soliton and the oscillatory wedge generated by the
continuous spectrum, including soliton transmission, trapping, and wake
formation on a modulationally unstable background.

The present paper addresses a different, but equally natural, question.  The
asymptotic formulae in~\cite{BM2017} describe the plane-wave and modulated
elliptic-wave regions away from the separating curves
\[
        x=\pm 4\sqrt{2}\,q_o t,
\]
but they are not uniform as one approaches these boundaries.  It is therefore
necessary to resolve the boundary layer between the plane-wave and
modulated elliptic-wave regimes; see Fig.~\ref{fig:regions}.  More precisely, we derive transition
asymptotics in the double-scaling regions
\[
        \left|\xi\pm 4\sqrt{2}\,q_o\right|\le C t^{-2/3},
\]
where \(C>0\) is fixed.  We show that, in each of these regions, the leading term of the solution is a plane wave.  The first nontrivial correction, however, is no longer the
standard \(t^{-1/2}\) correction that appears in the  plane-wave and
elliptic-wave regions.  Instead, it is of order \(t^{-1/3}\), and its
coefficient is expressed in terms of a distinguished tritronqu\'ee solution
\(\Q(y)\) of the inhomogeneous Painlev\'e-II equation
\be \label{E:NPII}
\frac{\mathrm d^2 \Q}{\mathrm d y^2}
+\frac{2}{3}y \Q
-2 \Q^3
+\frac{2}{3}\ii \nu
-\frac{1}{3}
=0,
\ee
characterized by the asymptotic condition
\be \label{E:CQ}
\Q(y)
=
\ii\left(-\frac{y}{3}\right)^{1/2}
-
\left(\frac{1}{4}-\frac{\ii \nu}{2}\right)\frac{1}{y}
+
\mathcal O\bigl(|y|^{-5/2}\bigr),
\qquad |y|\to\infty,\quad |\arg(-y)|<\frac{2\pi}{3}.
\ee
Here \(\nu\) is a constant determined by the reflection coefficient \(r(k)\)
at the critical point \(k_c=-q_o/\sqrt{2}\).

A systematic study of the relevant increasing tritronqu\'ee solutions of the inhomogeneous Painlev\'e-II equation was carried out by Miller~\cite{Miller}. In particular, Miller analyzed the RH representation of these solutions, their large-argument connection formulae, and the global behavior of the special solutions needed in applications. One important motivation came from the study of rogue waves of infinite order~\cite{DLM2020}. Bilman, Ling and Miller showed that high-order fundamental rogue waves of the focusing NLS equation have a nontrivial near-field limit, called the rogue wave of infinite order, which is itself a special solution of the focusing NLS equation and is related to the Painlev\'e-III hierarchy. In a transitional far-field regime of that limiting rogue wave, the asymptotics are described by a special Painlev\'e-II tritronqu\'ee solution.   

Since the left and right transition regions can be treated
in a similar way, we focus on the left transition region
\be \label{E:PL}
\cP=
\left\{(x,t)\in\R\times\R_+:\ |\xi+  4\sqrt{2}\,q_o  |\le Ct^{-2/3}\right\}.
\ee

The main theorem of this paper is stated as follows.

\begin{theorem}[Asymptotics in the transition region \(\mathcal P\)]\label{Th:main}
Let \(q(x,t)\) be the solution of the IVP~\eqref{ds-fnls-ivp} and assume that the
initial datum satisfies~\eqref{ds-fnls-ic-space} and  Assumption~\ref{As:gesa}.
Let \(r(k)\) denote the reflection coefficient associated with the initial
datum. Then, as \(t\to\infty\),
\be \label{E:asygs}
q(x,t)
=
q_- e^{2 \ii g_\infty}
+
t^{-1/3} q_{\rm p}(x,t)
+
\mathcal O\bigl(t^{-2/3}\log t\bigr),
\ee
uniformly for \((x,t)\in\mathcal P\). Here the phase \(g_\infty\) is defined
by~\eqref{E:defginfty}. The bounded function \(q_{\rm p}(x,t)\) is given by
\[
q_{\rm p}(x,t)
=
\frac{
2\ii\left[
\alpha_1 \mathcal V(y)+\alpha_2 \mathcal V^*(y)+\alpha_3 \widehat{\mathcal V}(y)
\right]
}{
\left(\tfrac{8\sqrt{6}}{9 q_o}\right)^{1/3}
},
\]
where the coefficients \(\{\alpha_j\}_{j=1}^3\) are defined
by~\eqref{E:defalpha}. The function \( \mathcal V (y)\) can be expressed in terms of the Painlev\'e-II tritronqu\'ee solution \(\mathcal Q(y)\) defined
in~\eqref{E:CQ} as follows:
\[
\mathcal V(y)=
\begin{cases}
\displaystyle
\alpha_0 \e^{\ii \phi_0}
e^{-\frac{2}{9} \sqrt{3}\ii \left(-y\right)^{3/2}}
(-3y)^{-\frac14+\frac{\ii \nu}{2}}
\\[1mm]
\displaystyle\qquad\qquad\qquad
{}\times
\exp\left\{
\int_{-\infty}^{y}
\left[
\Q(s)-\ii\left(-\frac{s}{3}\right)^{1/2}
+
\left(\frac14-\frac{\ii \nu}{2}\right)\frac{1}{s}
\right]\mathrm ds
\right\},
& y<0,
\\[5mm]
\displaystyle
\mathcal V(-1)\exp\left\{
\int_{-1}^{y}\Q(s)\,\mathrm ds
\right\},
& y\ge 0,
\end{cases}
\]
where
$$
y=
\frac{2}{\sqrt{3}}
\left(\frac{8\sqrt{6}}{9q_o}\right)^{-1/3}
t^{2/3}(\xi+ 4\sqrt{2}\,q_o ).
$$
Moreover,
\[
\alpha_0=\sqrt{-\frac{\nu}{2}},
\qquad
\phi_0=-\frac{3}{4}\pi+\nu\log 2+\arg\Gamma(-\ii\nu),
\]
and
\[
\nu=-\frac{1}{2\pi}\log\bigl(1+|r(k_c)|^2\bigr),
\qquad
k_c=-\frac{q_o}{\sqrt{2}}.
\]
Finally, \(\widehat{\mathcal V}(y)\) can be represented in terms of \(\mathcal V(y)\) by
\[
\widehat {\mathcal V}(y)=
\begin{cases}
\displaystyle
\ii\nu\left(-\frac{y}{3}\right)^{1/2}
+
\ii\int_{-\infty}^{y}
\left[
|\mathcal V (s)|^2
+
\frac{\nu}{2\sqrt{3}}(-s)^{-1/2}
\right]\,\mathrm ds,
& y<0,
\\[5mm]
\displaystyle
\widehat{\mathcal V}(-1)
+
\ii\int_{-1}^{y}|\mathcal V (s)|^2\,\mathrm ds,
& y\ge0.
\end{cases}
\]
\end{theorem}

 \begin{remark}
Throughout this paper, we assume that the reflection coefficient satisfies
\(
r(k_c)\neq 0
\),
which ensures that the Painlev\'e region is nondegenerate.
\end{remark}

\begin{remark}
For the Painlev\'e region on the other side, one may first rescale the
original RH problem; see~\cite[Eq.~(4.5)]{BM2017} and the related discussion
therein.  After this rescaling, the subsequent analysis is completely parallel
to the computations presented in this paper.  For brevity, we omit the
details. We emphasize, however, that the final asymptotic formula has the same form, with the relevant quantities modified accordingly.
\end{remark}

\noindent{\bf Organization of the paper.}
Section~\ref{s:rhch} recalls the inverse-scattering and RH formulation for the IVP~\eqref{ds-fnls-ivp}. The Painlev\'e transition region $\cP$ is divided into two subregions, \(\cP_-\) and \(\cP_+\), which are analyzed separately. Section~\ref{S:dza} treats the case \((x,t)\in\cP_-\), where we construct the local Painlev\'e-II parametrix and formulate the associated small-norm RH problem. Section~\ref{S:dzcpplus} explains the modifications needed for the case \((x,t)\in\cP_+\). Appendix~\ref{App:PII} collects the Painlev\'e-II model problems and the tritronqu\'ee asymptotics used in the proof.

\noindent{\bf Notation.}  Throughout this paper, the following notation will be used.
\begin{itemize}
   \item[$\bullet$]  The
symbols \(C>0\) and \(c>0\) denote generic constants whose values may change
from line to line.
\item[$\bullet$]
Unless otherwise stated, $\log (z)$ always denotes
the principal branch of the logarithm.
\item[$\bullet$]
   The asterisk denotes complex conjugation.  For a complex-valued function $f(k)$, we use 
	\begin{equation*}
		\bar{f}:=f^*(k^*), \qquad k\in\mathbb{C}.
	\end{equation*}
	\item[$\bullet$] We  set
	 \(\R_+=(0,+\infty)\), \(\R_-=(-\infty,0)\), and
denote the upper and lower half-planes by \(\C_+\) and \(\C_-\). 
	\item[$\bullet$]As usual, the classical Pauli matrices $\{\sigma_j\}_{j=1,2,3}$ are defined by
	\begin{equation}\label{def:PauliM}
		\sigma_1:=\begin{pmatrix}0 & 1 \\ 1 & 0\end{pmatrix}, \quad
		\sigma_2:=\begin{pmatrix}0 & -\ii \\ \ii & 0\end{pmatrix}, \quad
		\sigma_3:=\begin{pmatrix}1 & 0 \\ 0 & -1\end{pmatrix}.
	\end{equation}
For a scalar function \(f(z)\), we set
\(
f^{\sigma_3}:=\diag(f,f^{-1}).
\)
	\item[$\bullet$]For a
matrix-valued function on a contour, all \(L^p\)-norms are understood
entrywise. For any smooth oriented curve $\Sigma$, the Cauchy operator $\mathcal{C}$ on $\Sigma$ is defined  by
	\begin{align*}
		\mathcal{C}f(k)=\frac{1}{2\pi \ii}\int_{\Sigma}\frac{f(\zeta)}{\zeta-k}\, d\zeta, \qquad  k\in\mathbb{C}\setminus \Sigma.
	\end{align*}
	Given a matrix-valued function $f \in L^p(\Sigma)$, $1\leq  p<\infty$,
	\begin{align}\label{def:opCpm}
		\mathcal{C}_\pm f(k):=\lim_{z'\to k\in\Sigma}\frac{1}{2\pi \ii}\int_{\Sigma}\frac{f(\zeta)}{\zeta-z'}\, d\zeta
	\end{align}
   stands for the positive/negative (according to the orientation of $\Sigma$) non-tangential boundary value of $\mathcal{C}f$.
\end{itemize}

\section{A Riemann--Hilbert formulation}\label{s:rhch}
In this section, we review the IST of IVP~\eqref{ds-fnls-ivp}. Since these results have been well established in Ref.~\cite{BM2017}, we will omit the proof.

The focusing NLS equation~\eqref{ds-fnls-ivp-fnls-eq} is a
completely integrable system.  It can be written as the
compatibility condition
\[
X_t-T_x+[X,T]=0
\]
of the Lax pair~\cite{zs1971,APT2004}
\begin{equation}\label{ds-fnls-lp}
\Psi_x=X\Psi,
\qquad
\Psi_t=T\Psi,
\end{equation}
where \(\Psi=\Psi(x,t,k)\) is a \(2\times2\) matrix-valued function, and
\begin{equation}\label{ds-fnls-lp-XT}
X=\mathrm{i}k\sigma_3+ Q,
\qquad
T=-2\mathrm{i}k^2\sigma_3
+\mathrm{i}\sigma_3\left(Q_x-Q^2-q_o^2I\right)
-2k Q.
\end{equation}
Here \(k\in\mathbb C\),
\begin{equation}\label{ds-fnls-lp-Q}
Q=
\begin{pmatrix}
0&q\\
-\bar q&0
\end{pmatrix},
\qquad
I=
\begin{pmatrix}
1&0\\
0&1
\end{pmatrix}.
\end{equation}

Let
\[
X_\pm=\lim_{x\to\pm\infty}X(x,t,k),
\qquad
T_\pm=\lim_{x\to\pm\infty}T(x,t,k).
\]
The eigenvector matrix of \(X_\pm\) can be chosen as
\begin{equation}\label{eq:Epm}
\mathcal E_\pm(k)=
\begin{pmatrix}
1&
\dfrac{\mathrm{i}(\lambda-k)}{\overline{q}_\pm}
\\[1.2ex]
\dfrac{\mathrm{i}(\lambda-k)}{q_\pm}
&
1
\end{pmatrix},
\end{equation}
where
\begin{equation}\label{eq:lambda-def}
\lambda(k)=\left(k^2+q_o^2\right)^{1/2}.
\end{equation}
The corresponding eigenvalues are \(\pm\mathrm{i}\lambda(k)\), and
we take the branch cut to be
\begin{equation}\label{eq:branch-cut}
B=\mathrm{i}[-q_o,q_o].
\end{equation}
The function \(\lambda(k)\) is chosen to be single-valued in
\(\mathbb C\setminus B\), with boundary values on \(B\) taken from the right.
Equivalently,
\begin{equation}\label{eq:lambda-branch}
\lambda(k)=
\begin{cases}
\sqrt{k^2+q_o^2}, & k\in\mathbb R_+\cup B,\\[0.6ex]
-\sqrt{k^2+q_o^2}, & k\in\mathbb R_-,
\end{cases}
\end{equation}
where the square root denotes the principal branch of the real square root.
The Jost matrices are defined as simultaneous solutions of the Lax
pair~\eqref{ds-fnls-lp} satisfying the boundary conditions
\begin{equation}
\Psi_\pm(x,t,k)
=
\mathcal E_{\pm}(k)
e^{\ii\theta(\xi,k)t\sigma_3}
\bigl[I+o(1)\bigr],
\qquad x\to\pm\infty,
\label{eq:Jost}
\end{equation}
where
\[
\theta(\xi,k)=\lambda(k)(\xi-2k).
\]
Then we define the spectral coefficient \(a(k)\), together with its
Schwarz conjugate \(\bar a(k)\), by
\begin{equation}\label{ds-a-ab-def}
a(k)
=
\frac{\operatorname{Wr}\left[\Psi_{-1}(x,t,k),\Psi_{+2}(x,t,k)\right]}
{d(k)},
\qquad
\bar a(k)
=
\frac{\operatorname{Wr}\left[\Psi_{+1}(x,t,k),\Psi_{-2}(x,t,k)\right]}
{d(k)}.
\end{equation}
Here \(\operatorname{Wr}\) denotes the Wronskian determinant, and
\begin{equation}\label{ds-d-def}
d(k):=\frac{2\lambda(k)}{\lambda(k)+k}.
\end{equation}
The Wronskian determinants appearing in~\eqref{ds-a-ab-def} are independent of
\(x\) and \(t\). Consequently, both \(a\) and \(\bar a\) are functions of the
spectral parameter \(k\) only.

Following the idea of~\cite{BM2017}, we introduce the following sectionally
meromorphic matrix-valued function \(M\):
\begin{equation}\label{ds-mrh}
M(x,t,k)
=
\begin{cases}
\left(
\dfrac{\Psi_{+1}(x,t,k)}{\bar{a}(k)\,d(k)},
\ \Psi_{-2}(x,t,k)
\right)
e^{-\mathrm{i}\theta(\xi,k)t\sigma_3},
&
k\in\mathbb C_+\setminus B^+,
\\[2.2ex]
\left(
\Psi_{-1}(x,t,k),
\ \dfrac{\Psi_{+2}(x,t,k)}{a(k)d(k)}
\right)
e^{-\mathrm{i}\theta(\xi,k)t\sigma_3},
&
k\in\mathbb C_-\setminus B^-,
\end{cases}
\end{equation}
where
$$
B^+:= \ii\left[0, q_o\right], 
\quad
B^- := \ii\left[-q_o, 0\right].
$$

Following~\cite{BM2017}, we impose the following spectral assumption.

\begin{assume}\label{As:gesa}
Assume that
\[
a(k)\neq 0,
\qquad
k\in \mathbb C_-\cup\Sigma,
\]
where
\(
\Sigma=\mathbb R\cup B.
\)
\end{assume}

Under Assumption~\ref{As:gesa}, it was shown in~\cite{BM2017} that
\(M(x,t,k)\) is analytic for \(k\in\mathbb C\setminus\Sigma\) and has jumps
across \(\Sigma\). More precisely, \(M\) satisfies the following
RH problem:
\begin{subequations}\label{ds-rhp-cpam}
\begin{align}
M_+(x,t,k)&=M_-(x,t,k)V_1(x,t,k),
&& k\in\mathbb R,
\\
M_+(x,t,k)&=M_-(x,t,k)V_2(x,t,k),
&& k\in B^+,
\\
M_+(x,t,k)&=M_-(x,t,k)V_3(x,t,k),
&& k\in B^-,
\\
M(x,t,k)&=I+\mathcal O\left(\frac{1}{k}\right),
&& k\to\infty.
\end{align}
\end{subequations}
The jump matrices on the three components \(\mathbb R\), \(B^+\), and \(B^-\)
of the continuous spectrum \(\Sigma\) are given by
\begin{subequations}\label{ds-rhp-cpam-j}
\begin{align}
V_1(x,t,k)
&=
\begin{pmatrix}
\dfrac{1+r(k)\bar r(k)}{d(k)}
&
\bar r(k)e^{2\mathrm{i}\theta(\xi,k)t}
\\[1.2ex]
r(k)e^{-2\mathrm{i}\theta(\xi,k)t}
&
d(k)
\end{pmatrix},
\label{ds-rhp-cpam-j1}
\\[2ex]
V_2(x,t,k)
&=
\begin{pmatrix}
-\tfrac{\lambda(k)-k}{\mathrm{i}q_-}\,
\bar r(k)e^{2\mathrm{i}\theta(\xi,k)t}
&
\tfrac{2\lambda(k)}{\mathrm{i}\bar q_-}
\\[2ex]
\dfrac{\bar q_-}{2\mathrm{i}\lambda(k)}
\left[1+r(k)\bar r(k)\right]
&
-\tfrac{\lambda(k)+k}{\mathrm{i}\bar q_-}\,
r(k)e^{-2\mathrm{i}\theta(\xi,k)t}
\end{pmatrix},
\label{ds-rhp-cpam-j2}
\\[2ex]
V_3(x,t,k)
&=
\begin{pmatrix}
\tfrac{\lambda(k)+k}{\mathrm{i}q_-}\,
\bar r(k)e^{2\mathrm{i}\theta(\xi,k)t}
&
\tfrac{q_-}{2\mathrm{i}\lambda(k)}
\left[1+r(k)\bar r(k)\right]
\\[2ex]
\tfrac{2\lambda(k)}{\mathrm{i}q_-}
&
\tfrac{\lambda(k)-k}{\mathrm{i}\bar q_-}\,
r(k)e^{-2\mathrm{i}\theta(\xi,k)t}
\end{pmatrix}.
\label{ds-rhp-cpam-j3}
\end{align}
\end{subequations}

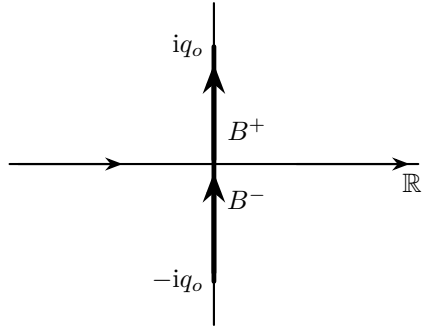
\begin{figure}[h]
\centering
\begin{tikzpicture}[
    scale=1.15,
    >=latex,
    line cap=round,
    line join=round,
    every node/.style={font=\small},
    axis/.style={draw=black, line width=0.8pt},
    contour/.style={draw=black, line width=1.8pt},
    arrowcontour/.style={
        contour,
        ->,
        >=Stealth
    },
    realarrow/.style={
        draw=black,
        line width=0.8pt,
        ->,
        >=Stealth
    }
]

\coordinate (O) at (0,0);
\coordinate (Iq) at (0,1.35);
\coordinate (NIq) at (0,-1.35);

\draw[axis] (-2.35,0) -- (2.35,0);
\draw[axis] (0,-1.85) -- (0,1.85);

\draw[realarrow] (-2.25,0) -- (-1.05,0);
\draw[realarrow] (0.95,0) -- (2.25,0);

\draw[contour] (0,-1.35) -- (0,1.35);

\draw[arrowcontour] (0,0.05) -- (0,1.15);
\draw[arrowcontour] (0,-1.25) -- (0,-0.10);

\node at (0.38,0.40) {$B^+$};
\node at (0.38,-0.40) {$B^-$};

\node at (2.28,-0.22) {$\mathbb R$};

\node[left] at (0,1.35) {$\mathrm{i}q_o$};
\node[left] at (0,-1.35) {$-\mathrm{i}q_o$};

\end{tikzpicture}
\caption{The contour \(\Sigma=\mathbb R\cup B\).}
\label{fig:continuous-spectrum-jumps}
\end{figure}
\noindent See Figure~\ref{fig:continuous-spectrum-jumps} for the orientation of these contours.
The reflection coefficient \(r\) is defined by
\begin{equation}\label{ds-r-coef-def}
r(k)
=
-\frac{b(k)}{\bar a(k)},
\qquad
b(k)
:=
\frac{
\operatorname{Wr}\left[\Psi_{+1}(x,t,k),\Psi_{-1}(x,t,k)\right]
}
{d(k)}.
\end{equation}
By~\cite[Lemma~3.1]{BM2017}, for the exponentially decaying initial data under
consideration, the reflection coefficient can be analytically continued to a
small neighborhood of the continuous spectrum \(\Sigma\).  This property provides the analytic foundation for the contour deformations carried out in
the following sections.

The \(x\)-part of the Lax pair~\eqref{ds-fnls-lp}, together with the
definition of \(M\) in~\eqref{ds-mrh} and the normalization condition, yields
the solution of the IVP~\eqref{ds-fnls-ivp} through the reconstruction formula
\begin{equation}\label{ds-q-recon-n}
q(x,t)
=
-2\mathrm{i}\lim_{k\to\infty}kM_{12}(x,t,k).
\end{equation}
Therefore, the long-time asymptotic behavior of the solution \(q\) of the
focusing NLS IVP~\eqref{ds-fnls-ivp} can be obtained equivalently by analyzing
the corresponding long-time behavior of the solution \(M\) of the
RH problem~\eqref{ds-rhp-cpam}.

\section{Long--time asymptotics in the transition region \texorpdfstring{$\cP_-$}{P-}}\label{S:dza}

It is convenient to divide the Painlev\'e region $\cP$ into
two subregions and treat them separately. We write
\[
\cP=\cP_+\cup\cP_-,
\]
where
\begin{equation}\label{E:Ppm}
\cP_+
:=
\cP\cap\{\xi\ge -4\sqrt{2}\,q_o\},
\qquad
\cP_-
:=
\cP\cap\{\xi\le -4\sqrt{2}\,q_o\}.
\end{equation}
In this section, we analyze the long-time behavior of the RH
problem~\eqref{ds-rhp-cpam} in the  region \(\cP_-\) by means of the
Deift--Zhou nonlinear steepest descent method.  The main idea of the Deift--Zhou steepest descent analysis is to transform the
original RH problem, through a sequence of exact and invertible
transformations, into a solvable model RH problem together with
an error-estimate problem. To construct these transformations, we first need to analyze the phase function
\(\theta(\xi,k)\).

Let \(\xi_c=-4\sqrt{2}\,q_o\). For $\xi < \xi_c$, the two stationary points of the phase function
\(\theta(\xi,k)\) are given by
\begin{equation}
 k_1(\xi)
 =
 \frac{1}{8}\left(\xi-\sqrt{\xi^2-\xi_c^2}\right),
 \qquad
 k_2(\xi)
 =
 \frac{1}{8}\left(\xi+\sqrt{\xi^2-\xi_c^2}\right).
\label{eq:stationary}
\end{equation}
When \(\xi\le \xi_c\), both stationary points lie on the real axis. At the
critical value \(\xi=\xi_c\), they coalesce at
\[
k_c=-\frac{q_o}{\sqrt{2}}.
\]
For \(\xi_c<\xi<0\), the two stationary points move off the real axis and form
a complex conjugate pair.  The sign distribution of \(\mathrm{Re}(\ii\theta)\)  is shown in Figure~\ref{sign-structure-f}.

\begin{figure}[h]
\centering
\includegraphics[width=0.31\textwidth]{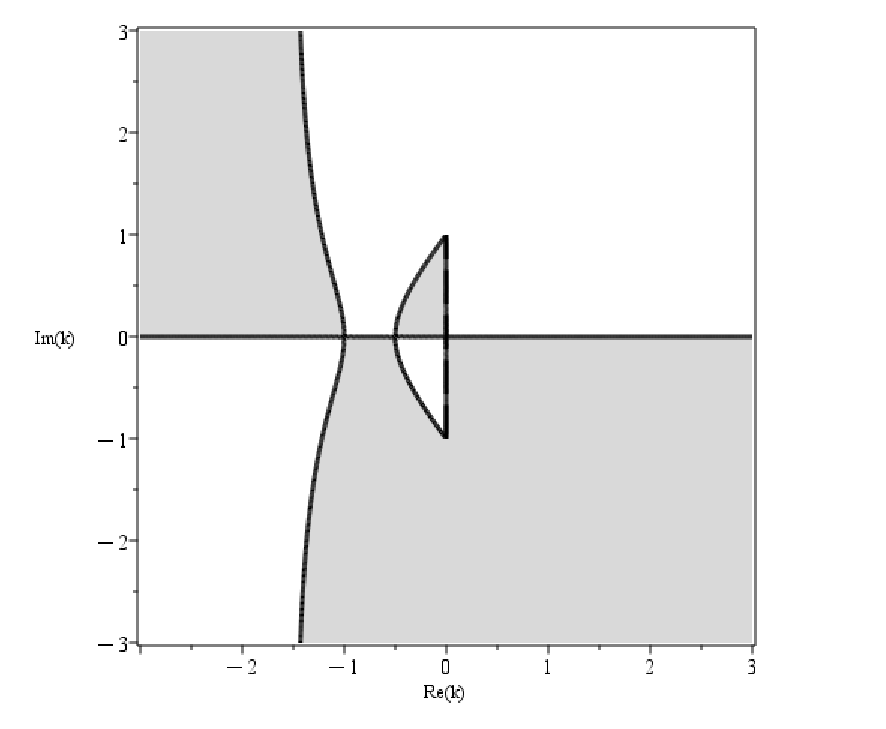}
\hfill
\includegraphics[width=0.31\textwidth]{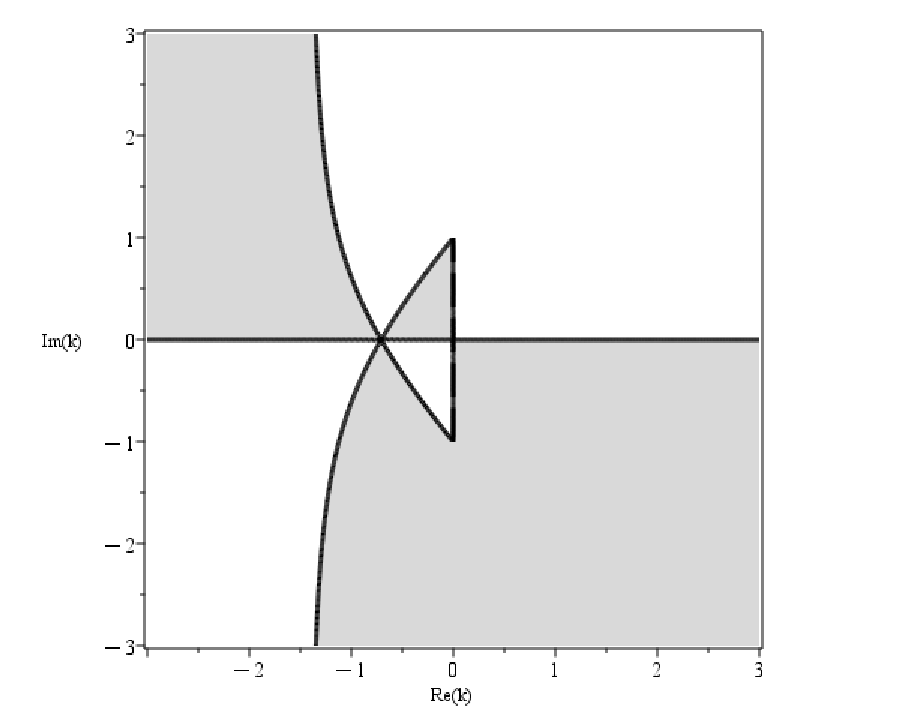}
\hfill
\includegraphics[width=0.31\textwidth]{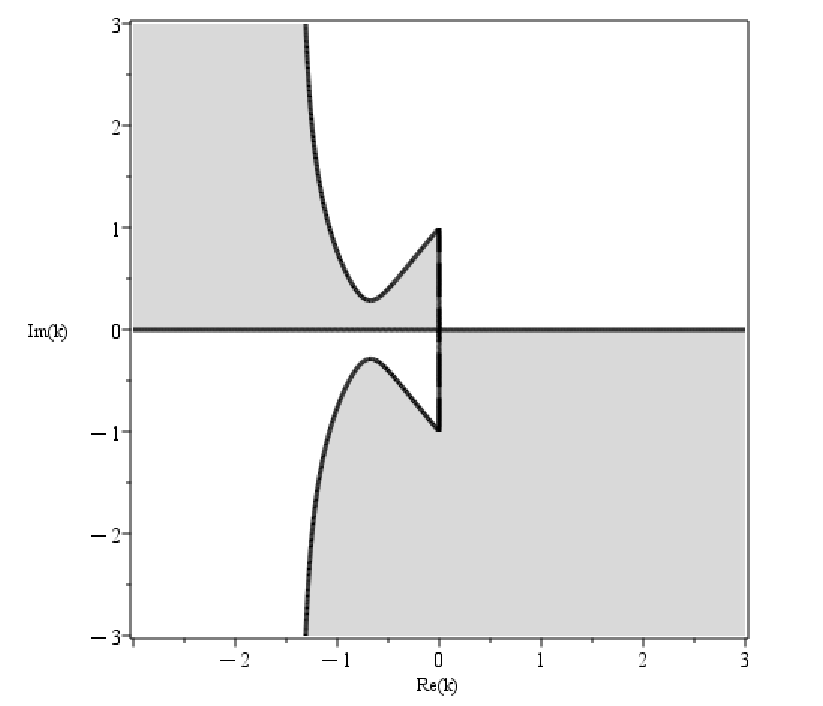}
\caption{
Sign structure of \(\mathrm{Re}(\mathrm{i}\theta)\) for
$q_o=1$, \(\xi=-6\), \(\xi=\xi_c\), and \(\xi=-5.5\), respectively.
The gray regions indicate \(\mathrm{Re}(\mathrm{i}\theta)<0\), while the
white regions indicate \(\mathrm{Re}(\mathrm{i}\theta)>0\).
}
\label{sign-structure-f}
\end{figure}

\subsection{Opening of lenses}

The purpose of this subsection is to carry out a sequence of preliminary
deformations of the original RH problem.  These transformations are very close
to those used in the plane-wave region, with only minor modifications.  We
therefore present the construction rather briefly and refer the reader
to~\cite[Section~4]{BM2017} for the corresponding standard argument.

The exponential decay assumption imposed on the initial datum in~\eqref{ds-fnls-ic-space}
ensures that the reflection coefficient admits an analytic continuation to a
small neighborhood of the continuous spectrum; see, for example,
\cite[Lemma~3.1]{BM2017}.  Throughout the following deformations, all contours
obtained by moving off the continuous spectrum are chosen inside this
neighborhood.

\vskip 2mm
\noindent\textbf{First deformation.} 
For \(\xi<\xi_c\), we first use the following factorizations of the jump
matrix on the real axis. More precisely,
\[
V_1
=
\begin{cases}
V_2^{(1)}V_0^{(1)}V_1^{(1)}, & \operatorname{Re} k<k_1,\\[0.6ex]
V_4^{(1)}V_3^{(1)}, & \operatorname{Re} k>k_c,
\end{cases}
\]
where
\[
V_0^{(1)}
=
\begin{pmatrix}
1+r\bar r & 0\\[0.6ex]
0 & \dfrac{1}{1+r\bar r}
\end{pmatrix},
\]
\[
V_1^{(1)}
=
\begin{pmatrix}
d^{-1/2}
&
\dfrac{d^{1/2}\bar r\,e^{2\mathrm{i}\theta t}}{1+r\bar r}
\\[1.2ex]
0
&
d^{1/2}
\end{pmatrix},
\qquad
V_2^{(1)}
=
\begin{pmatrix}
d^{-1/2}
&
0
\\[1.2ex]
\dfrac{d^{1/2}r\,e^{-2\mathrm{i}\theta t}}{1+r\bar r}
&
d^{1/2}
\end{pmatrix},
\]
\[
V_3^{(1)}
=
\begin{pmatrix}
d^{-1/2}
&
0
\\[1.2ex]
d^{-1/2}r\,e^{-2\mathrm{i}\theta t}
&
d^{1/2}
\end{pmatrix},
\qquad
V_4^{(1)}
=
\begin{pmatrix}
d^{-1/2}
&
d^{-1/2}\bar r\,e^{2\mathrm{i}\theta t}
\\[1.2ex]
0
&
d^{1/2}
\end{pmatrix}.
\]
We then introduce the first transformation as follows:
\be  \label{E:firsttrans}
M^{(1)}(x,t,k)
=
M(x,t,k)
\begin{cases}
\bigl(V_3^{(1)}\bigr)^{-1}, & k\in \mathcal R_1^{(1)},\\[0.6ex]
\bigl(V_1^{(1)}\bigr)^{-1}, & k\in \mathcal R_2^{(1)},\\[0.6ex]
V_2^{(1)}, & k\in \mathcal R_3^{(1)},\\[0.6ex]
V_4^{(1)}, & k\in \mathcal R_4^{(1)},
\end{cases}
\ee
where the regions \(\{\mathcal R_j^{(1)}\}_{j=1}^4\) are those shown in
Figure~\ref{F:Sigma1}.

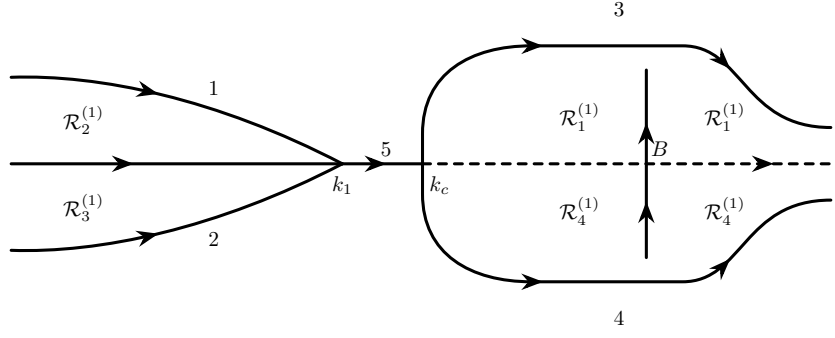
\begin{figure}[htbp]
\centering
\begin{tikzpicture}[
    x=1cm,
    y=1cm,
    scale=0.8,
    transform shape,
    flow/.style={
        draw=black,
        line width=1.15pt,
        line cap=round,
        line join=round,
        postaction={decorate},
        decoration={
            markings,
            mark=at position #1 with {\arrow{Stealth}}
        }
    },
    axisflow/.style={
        draw=black,
        dashed,
        line width=1.15pt,
        line cap=round,
        line join=round,
        postaction={decorate},
        decoration={
            markings,
            mark=at position #1 with {\arrow{Stealth}}
        }
    },
    twoflow/.style args={#1 and #2}{
        draw=black,
        line width=1.15pt,
        line cap=round,
        line join=round,
        postaction={decorate},
        decoration={
            markings,
            mark=at position #1 with {\arrow{Stealth}},
            mark=at position #2 with {\arrow{Stealth}}
        }
    },
    regionlabel/.style={font=\small},
    contourlabel/.style={font=\small},
    pointlabel/.style={font=\small}
]

  \coordinate (Kone) at (-0.625,0);
  \coordinate (Kc)   at ( 0.70,0);

  \draw[flow=.43]
    (-6.10, 1.43)
      .. controls (-4.65, 1.48) and (-2.70, 1.08) .. (Kone);
  \draw[flow=.43]
    (-6.10,-1.43)
      .. controls (-4.65,-1.48) and (-2.70,-1.08) .. (Kone);

  \draw[flow=.37] (-6.10,0) -- (Kone);
  \draw[flow=.55] (Kone) -- (Kc);

  \draw[axisflow=.86] (Kc) -- (7.45,0);

  \draw[draw=black,line width=1.15pt,line cap=round] (Kc)--(0.70,0.50);
  \draw[draw=black,line width=1.15pt,line cap=round] (Kc)--(0.70,-0.50);
  \draw[twoflow=.35 and .75]
    (0.70,0.50)
      .. controls (0.70, 1.5) and (1.5, 1.95) .. (2.5, 1.95)
      .. controls (3.5, 1.95) .. (5, 1.95)
      .. controls (6, 1.95) and (6, 0.6) .. (7.45, 0.6);
  \draw[twoflow=.35 and .75]
    (0.70,-0.50)
      .. controls (0.70, -1.5) and (1.5, -1.95) .. (2.5, -1.95)
      .. controls (3.5, -1.95) .. (5, -1.95)
      .. controls (6, -1.95) and (6, -0.6) .. (7.45, -0.6);

  \draw[twoflow=.30 and .72] (4.40,-1.55) -- (4.40,1.55);

  \node[regionlabel] at (-4.90, 0.72) {$\mathcal{R}^{(1)}_2$};
  \node[regionlabel] at (-4.90,-0.72) {$\mathcal{R}^{(1)}_3$};
  \node[regionlabel] at ( 3.30, 0.78) {$\mathcal{R}^{(1)}_1$};
  \node[regionlabel] at ( 5.70, 0.78) {$\mathcal{R}^{(1)}_1$};
  \node[regionlabel] at ( 3.30,-0.78) {$\mathcal{R}^{(1)}_4$};
  \node[regionlabel] at ( 5.70,-0.78) {$\mathcal{R}^{(1)}_4$};

  \node[contourlabel] at (-2.75, 1.24) {$1$};
  \node[contourlabel] at (-2.75,-1.24) {$2$};
  \node[contourlabel] at ( 3.95, 2.55) {$3$};
  \node[contourlabel] at ( 3.95,-2.55) {$4$};
  \node[contourlabel] at ( 0.10, 0.24) {$5$};
  \node[contourlabel] at ( 4.62, 0.25) {$B$};

  \node[pointlabel,below=2pt] at (Kone) {$k_1$};
  \node[pointlabel,below=2pt,xshift=8pt] at (Kc) {$k_c$};

\end{tikzpicture}
\caption{The jump contour \(\Sigma^{(1)}\) for $M^{(1)}$ and the regions \(\{\mathcal R_j^{(1)}\}_{j=1}^4\).}
\label{F:Sigma1}
\end{figure}

 We denote the new jump contour by
\(\Sigma^{(1)}\), and write the corresponding jump matrices as \(V^{(1)}\).
On each component \(\Sigma_j^{(1)}\), the jump matrix is denoted by
\(V_j^{(1)}\). The explicit expressions for \(V_j^{(1)}\), \(j=1,2,3,4\), are
given above. On the branch cut \(B\), the jump matrix is
\be \label{E:defVB}
V_B^{(1)}=V_B
=
\begin{pmatrix}
0 & \dfrac{q_-}{\mathrm{i}q_o}\\[1.2ex]
\dfrac{\bar q_-}{\mathrm{i}q_o} & 0
\end{pmatrix}.
\ee
On $(-\infty,k_1)$, $V^{(1)}(k)=V^{(1)}_0$. Moreover, on the remaining part of the real axis  we have
\(
V^{(1)}_5(k)=V_1(k).
\)

\vskip 2mm
\noindent\textbf{Second deformation.} 
The second transformation is designed to remove the jump on
\((-\infty,k_1)\). Define the scalar function \(\delta\) by
\[
\delta_+(k)
=
\delta_-(k)\bigl(1+r(k)\bar r(k)\bigr),
\qquad
k\in(-\infty,k_c),
\]
together with the normalization
\[
\delta(k)=1+ \mathcal O(k^{-1}),
\qquad
k\to\infty.
\]
The Plemelj formulae give the explicit representation 
\be \label{ds-delta-def}
\delta( k) = \exp\left\{\frac{1}{2\ii\pi}\int_{-\infty}^{k_c} \frac{\log \left[1+r(s)\bar r(s)\right]}{s-k} \, ds\right\},
\quad k\notin (-\infty, k_c).
\ee
We then set
\[
M^{(2)}(x,t,k)
=
M^{(1)}(x,t,k)\delta(k)^{-\sigma_3},
\qquad
k\in\mathbb C.
\]

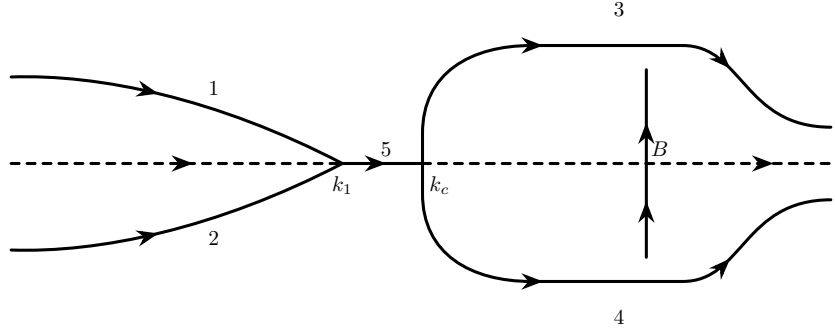
\begin{figure}[htbp]
\centering
\begin{tikzpicture}[
    x=1cm,
    y=1cm,
    scale=0.8,
    transform shape,
    flow/.style={
        draw=black,
        line width=1.15pt,
        line cap=round,
        line join=round,
        postaction={decorate},
        decoration={
            markings,
            mark=at position #1 with {\arrow{Stealth}}
        }
    },
    axisflow/.style={
        draw=black,
        dashed,
        line width=1.15pt,
        line cap=round,
        line join=round,
        postaction={decorate},
        decoration={
            markings,
            mark=at position #1 with {\arrow{Stealth}}
        }
    },
    twoflow/.style args={#1 and #2}{
        draw=black,
        line width=1.15pt,
        line cap=round,
        line join=round,
        postaction={decorate},
        decoration={
            markings,
            mark=at position #1 with {\arrow{Stealth}},
            mark=at position #2 with {\arrow{Stealth}}
        }
    },
    regionlabel/.style={font=\small},
    contourlabel/.style={font=\small},
    pointlabel/.style={font=\small}
]

  \coordinate (Kone) at (-0.625,0);
  \coordinate (Kc)   at ( 0.70,0);

  \draw[flow=.43]
    (-6.10, 1.43)
      .. controls (-4.65, 1.48) and (-2.70, 1.08) .. (Kone);
  \draw[flow=.43]
    (-6.10,-1.43)
      .. controls (-4.65,-1.48) and (-2.70,-1.08) .. (Kone);

  \draw[axisflow=.55] (-6.10,0) -- (Kone);
  \draw[flow=.55] (Kone) -- (Kc);

  \draw[axisflow=.86] (Kc) -- (7.45,0);

  \draw[draw=black,line width=1.15pt,line cap=round] (Kc)--(0.70,0.50);
  \draw[draw=black,line width=1.15pt,line cap=round] (Kc)--(0.70,-0.50);
  \draw[twoflow=.35 and .75]
    (0.70,0.50)
      .. controls (0.70, 1.5) and (1.5, 1.95) .. (2.5, 1.95)
      .. controls (3.5, 1.95) .. (5, 1.95)
      .. controls (6, 1.95) and (6, 0.6) .. (7.45, 0.6);
  \draw[twoflow=.35 and .75]
    (0.70,-0.50)
      .. controls (0.70, -1.5) and (1.5, -1.95) .. (2.5, -1.95)
      .. controls (3.5, -1.95) .. (5, -1.95)
      .. controls (6, -1.95) and (6, -0.6) .. (7.45, -0.6);

  \draw[twoflow=.30 and .72] (4.40,-1.55) -- (4.40,1.55);

  \node[contourlabel] at (-2.75, 1.24) {$1$};
  \node[contourlabel] at (-2.75,-1.24) {$2$};
  \node[contourlabel] at ( 3.95, 2.55) {$3$};
  \node[contourlabel] at ( 3.95,-2.55) {$4$};
  \node[contourlabel] at ( 0.10, 0.24) {$5$};
  \node[contourlabel] at ( 4.62, 0.25) {$B$};

  \node[pointlabel,below=2pt] at (Kone) {$k_1$};
  \node[pointlabel,below=2pt,xshift=8pt] at (Kc) {$k_c$};

\end{tikzpicture}
\caption{The jump contour \(\Sigma^{(2)}\) for \(M^{(2)}\).}
\label{F:Sigma2}
\end{figure}
Let \(V^{(2)}\) denote the corresponding jump matrix. With this transformation,
the jump on \((-\infty,k_1)\) is removed. The new jump contour is shown in
Figure~\ref{F:Sigma2}. The remaining jump matrices are given by
\begin{align*}
V_B^{(2)}
&=
\begin{pmatrix}
0 & \dfrac{q_-}{\mathrm{i}q_o}\delta^2 \\[1ex]
\dfrac{\bar q_-}{\mathrm{i}q_o}\delta^{-2} & 0
\end{pmatrix},
\qquad
V_1^{(2)}
=
\begin{pmatrix}
d^{-\frac12}
&
\dfrac{d^{\frac12}\bar r e^{2\mathrm{i}\theta t}}{1+r\bar r}\delta^2
\\[1.2ex]
0
&
d^{\frac12}
\end{pmatrix},
\\[2ex]
V_2^{(2)}
&=
\begin{pmatrix}
d^{-\frac12}
&
0
\\[1.2ex]
\dfrac{d^{\frac12}r e^{-2\mathrm{i}\theta t}}{1+r\bar r}\delta^{-2}
&
d^{\frac12}
\end{pmatrix},
\qquad
V_3^{(2)}
=
\begin{pmatrix}
d^{-\frac12}
&
0
\\[1.2ex]
d^{-\frac12}r e^{-2\mathrm{i}\theta t}\delta^{-2}
&
d^{\frac12}
\end{pmatrix},
\\[2ex]
V_4^{(2)}
&=
\begin{pmatrix}
d^{-\frac12}
&
d^{-\frac12}\bar r e^{2\mathrm{i}\theta t}\delta^2
\\[1.2ex]
0
&
d^{\frac12}
\end{pmatrix},
\qquad
V_5^{(2)}
=
\delta_-^{\sigma_3}V_5^{(1)}\delta_+^{-\sigma_3}.
\end{align*}

\vskip 2mm
\noindent\textbf{Third deformation.}
We eliminate the factor $d(k)$ from the jump matrices by introducing a new unknown $M^{(3)}$ in terms of $M^{(2)}$. Define
\[
\mathcal R_1^{(3)}
:=
\mathbb C_+
\setminus
\bigl(\mathcal R_1^{(1)}\cup\mathcal R_2^{(1)}\bigr),
\qquad
\mathcal R_2^{(3)}
:=
\mathbb C_-
\setminus
\bigl(\mathcal R_3^{(1)}\cup\mathcal R_4^{(1)}\bigr).
\]
We set
\[
M^{(3)}(x,t,k)
=
M^{(2)}(x,t,k)D(k)=M^{(2)}(x,t,k) \times
\begin{cases}
d(k)^{\sigma_3/2}, & k\in \mathcal R_1^{(3)},\\[0.6ex]
d(k)^{-\sigma_3/2}, & k\in \mathcal R_2^{(3)}.
\end{cases}
\]
The jump contour \(\Sigma^{(3)}\) remains unchanged. We denote the corresponding
jump matrix by \(V^{(3)}\). Its explicit form will be recorded below.

\begin{equation}
\begin{gathered}
V_1^{(3)}=
\begin{pmatrix}
1&\dfrac{\bar r\e^{2\ii\theta t}}{1+r\bar r}\,\delta^2\\[3mm]
0&1
\end{pmatrix},
\qquad
V_2^{(3)}=
\begin{pmatrix}
1&0\\[2mm]
\dfrac{r\e^{-2\ii\theta t}}{1+r\bar r}\,\delta^{-2}&1
\end{pmatrix},\\[3mm]
V_3^{(3)}=
\begin{pmatrix}
1&0\\[1mm]
r\e^{-2\ii\theta t}\delta^{-2}&1
\end{pmatrix},
\qquad
V_4^{(3)}=
\begin{pmatrix}
1&\bar r\e^{2\ii\theta t}\delta^2\\[1mm]
0&1
\end{pmatrix},\\[3mm]
V_5^{(3)} = \begin{pmatrix} 1 & \tfrac{\bar r e^{2\mathrm{i}\theta t}}{1+r\bar r}  (\delta_+)^2 \\[1.4ex] \tfrac{r  e^{-2\mathrm{i}\theta t} }{1+r\bar r} \tfrac{1}{(\delta_-)^2} & 1+r\bar r \end{pmatrix},\qquad
V_B^{(3)}=
\begin{pmatrix}
0&\dfrac{q_-}{\ii q_o}\,\delta^2\\[2mm]
\dfrac{\bar q_-}{\ii q_o}\,\delta^{-2}&0
\end{pmatrix}.
\end{gathered}  \label{eq:third-jumps}
\end{equation}

\vskip 2mm
\noindent\textbf{Fourth deformation.}
Our final goal is to convert the jump along the branch cut $B$ into the constant matrix $V_B$ given by \eqref{E:defVB}. This can be achieved by means of the global transformation
\be \label{ds-esc-n4-def-pw1}
M^{(4)}(x, t, k) = M^{(3)}(x, t, k) e^{\ii g( k)\sigma_3},
\ee
where the function $g(k)$ is analytic in $\mathbb C\setminus B$ and satisfies the jump condition
\be \label{ds-esc-jump-g-pw1}
 e^{\ii(g^+ + g^-)} = \delta^2,\quad k\in B,
\ee
and the normalization condition
\be \label{ds-esc-g-rhp-pw1}
\frac{g(k)}{\lambda(k)}
=
\mathcal O\left(\frac 1k\right), \quad k \to\infty.
\ee
Indeed, the jump condition \eqref{ds-esc-jump-g-pw1} implies that the jump of $M^{(4)}$ along $B$  is precisely $V_B$. Equations \eqref{ds-esc-jump-g-pw1} and \eqref{ds-esc-g-rhp-pw1} formulate a RH problem for $g$, which can be solved explicitly to yield 
\be
 \label{ds-esc-g-def-pw1}
g( k)
=
\frac{\lambda(k)}{2 \ii \pi^2}
\int_{\zeta\in B} \frac{1}{\lambda(\zeta)\left(\zeta-k\right)} \int_{-\infty}^{k_c} \frac{\log  \left[1+r(s)\bar r(s)\right]}{s-\zeta}\, ds d\zeta, \quad k\notin B.
\ee
Notice that, after the above transformations, \(M^{(4)}\) is no longer
normalized to the identity as \(k\to\infty\). Instead, we have
\[
M^{(4)}(x,t,k)
=
\left[I+ \mathcal O\left(k^{-1}\right)\right]e^{\mathrm{i}g_\infty\sigma_3},
\qquad
k\to\infty,
\]
where
\be \label{E:defginfty}
g_\infty
:=
\lim_{k\to\infty}g(k)
=
-\frac{1}{2\mathrm{i}\pi^2}
\int_{\zeta \in B}
\frac{1}{\lambda(\zeta)}
\int_{-\infty}^{k_c}
\frac{\log\left(1+|r(s)|^2\right)}{s-\zeta}
\,ds\,d\zeta .
\ee
A direct calculation shows that
\(
g_\infty\in\mathbb R.
\)

The jump contour for  \(M^{(4)}\) remains unchanged. The jump matrix $V^{(4)}$ is as follows:
\begin{align*}
V_1^{(4)}&=
\begin{pmatrix}
1&\dfrac{\bar r\e^{2 \ii(\theta t-g)}}{1+r\bar r}\,\delta^2\\[3mm]
0&1
\end{pmatrix},
\qquad
V_2^{(4)}=
\begin{pmatrix}
1&0\\[2mm]
\dfrac{r\e^{-2 \ii(\theta t-g)}}{1+r\bar r}\,\delta^{-2}&1
\end{pmatrix}, \\[3mm]
V_3^{(4)}&=
\begin{pmatrix}
1&0\\[1mm]
r\e^{-2 \ii(\theta t-g)}\delta^{-2}&1
\end{pmatrix},
\qquad
V_4^{(4)}=
\begin{pmatrix}
1&\bar r\e^{2 \ii(\theta t-g)}\delta^2\\[1mm]
0&1
\end{pmatrix},\\[3mm]
V_5^{(4)}& = \begin{pmatrix} 1 & \tfrac{\bar r e^{2\mathrm{i}(\theta t-g)}}{1+r\bar r}  (\delta_+)^2 \\[1.4ex] \tfrac{r  e^{-2\mathrm{i}(\theta t -g) } }{1+r\bar r} \tfrac{1}{(\delta_-)^2} & 1+r\bar r \end{pmatrix},\qquad
V_B^{(4)}=
\begin{pmatrix}
0 & \dfrac{q_-}{\mathrm{i}q_o}\\[1.2ex]
\dfrac{\bar q_-}{\mathrm{i}q_o} & 0
\end{pmatrix}.
\end{align*}

\subsection{The local parametrix and outer parametrix}\label{sub:lG}

In the previous subsection, after a sequence of contour deformations, all jump matrices were made exponentially close to the identity as \(t\to\infty\), except for the jumps on the branch cut and the jumps in a small neighborhood of the critical point \(k_c\). Therefore, two model constructions are needed. First, we construct a local parametrix in a neighborhood of \(k_c\).  Second, we construct an outer parametrix which absorbs the nontrivial jump on the branch cut.

\vskip 2mm
\noindent\textbf{Local parametrix.} 
We first record the Taylor expansion of the phase function near the critical
point. Recall that
\[
\xi_c=-4\sqrt{2}\,q_o,\qquad
k_c=-\frac{q_o}{\sqrt{2}}.
\]
Then, as \(k\to k_c\),
\[
\begin{aligned}
\theta(\xi,k)
={}&
\theta(\xi,k_c)
+\frac{\xi+4\sqrt{2}\,q_o}{\sqrt{3}}(k-k_c)
+\frac{4\sqrt{6}}{9q_o}(k-k_c)^3        \\
&\quad
-\frac{\sqrt{6}}{9q_o}(\xi+4\sqrt{2}\,q_o)(k-k_c)^2
-\frac{2\sqrt{3}}{27q_o^2}(\xi+4\sqrt{2}\,q_o)(k-k_c)^3
+\mathcal O\bigl((k-k_c)^4\bigr),
\end{aligned}
\]
where 
\be \label{E:thetaxikc}
\theta(\xi,k_c)
=
3\sqrt{3}\,q_o^2
-q_o\sqrt{\frac{3}{2}}\,(\xi-\xi_c).
\ee
Equivalently,
\be
\theta(\xi,k)
=
\theta(\xi,k_c)
+\frac{\xi-\xi_c}{\sqrt{3}}(k-k_c)
+\frac{4\sqrt{6}}{9q_o}(k-k_c)^3
+S(\xi,k),
\ee
where
\be \label{E:Sdef}
S(\xi,k)
=
-\frac{\sqrt{6}}{9q_o}(\xi-\xi_c)(k-k_c)^2
-\frac{2\sqrt{3}}{27q_o^2}(\xi-\xi_c)(k-k_c)^3
+\mathcal O\bigl((k-k_c)^4\bigr).
\ee
We introduce the scaled variable \(z\) and $y$ by 
\be \label{E:yz}
\begin{aligned}
z&=
\left(\frac{8\sqrt{6}}{9q_o}\right)^{1/3}
t^{1/3}(k-k_c),\\
y
&=
\frac{2}{\sqrt{3}}
\left(\frac{8\sqrt{6}}{9q_o}\right)^{-1/3}
t^{2/3}(\xi-\xi_c).
\end{aligned}
\ee
Let $\cD_{\eps}$ denote the open disk of radius $\eps$ centered at the
point $k_c$.  Then $k \to z$ is a biholomorphism from  $\cD_{\eps}$ onto the open disk of radius  $ \left(\frac{8\sqrt{6}}{9q_o}\right)^{1/3}
t^{1/3} \eps $  centered at the origin. 
With these definitions, 
we obtain
\[
\begin{aligned}
2t\bigl[\theta(\xi,k)-\theta(\xi,k_c)\bigr]
&=
yz+z^3+2tS(\xi,k), \qquad k \in \cD_{\eps}.
\end{aligned}
\]
Next we derive a  local logarithmic representation of the scalar function $\delta(z)$.
Near \(k=k_c\), the local logarithmic representation is
\[
\delta(k)
=
(k-k_c)^{\mathrm{i}\nu}
e^{\chi(k)}, \qquad k \in \cD_{\eps},
\]
where
\be \label{E:nu}
\nu
=
-\frac{1}{2\pi}
\log \left(1+|r(k_c)|^2\right),
\ee
and
\[
\chi(k)
=
-\frac{1}{2\pi\mathrm{i}}
\int_{-\infty}^{k_c}
\log(k-s)\,
\mathrm{d}\log\left(1+|r(s)|^2\right).
\]
Using
\[
k-k_c
=
\left(\frac{8\sqrt{6}}{9q_o}\right)^{-1/3}
t^{-1/3}z,
\]
we have
\[
\begin{aligned}
\delta(k)=
\exp\left\{
-\frac{\mathrm{i}\nu}{3}
\log\left(\frac{8\sqrt{6}}{9q_o}\right)
-\frac{\mathrm{i}\nu}{3}\log t
\right\}
e^{\chi(k)}
e^{\mathrm{i}\nu \log z}.
\end{aligned}
\]

We then decompose the factor \(\delta^2 e^{-2\mathrm{i}g}\) near the critical
point \(k=k_c\). Using the local logarithmic representation of \(\delta\), for $k \in \cD_{\eps}$ we
obtain
\[
\begin{aligned}
\delta^2 e^{-2\mathrm{i}g}
&=
e^{2\mathrm{i}\nu\log z}
\left(
\frac{8\sqrt{6}}{9q_o}t
\right)^{-\frac{2\mathrm{i}\nu}{3}}
e^{2\chi(k_c)}
e^{-2\mathrm{i}g(k_c)}
\left[
e^{2(\chi(k)-\chi(k_c))}
e^{-2\mathrm{i}g(k)+2\mathrm{i}g(k_c)}
\right]  =
z^{2\mathrm{i}\nu}d_0d_1,
\end{aligned}
\]
where
\[
d_0
=
\left(
\frac{8\sqrt{6}}{9q_o}t
\right)^{-\frac{2\mathrm{i}\nu}{3}}
\exp\left\{
2\chi(k_c)-2\mathrm{i}g(k_c)
\right\},
\]
and
\[
d_1
=
\exp\bigl\{
2\bigl(\chi(k)-\chi(k_c)\bigr)
-2\mathrm{i}\bigl(g(k)-g(k_c)\bigr)
\bigr\}.
\]

We now introduce the local conjugating matrix \(Y\). Let
\be \label{E:r0}
r_0=|r(k_c)|,
\qquad
r(k_c)=r_0e^{\mathrm{i}\arg r(k_c)}.
\ee
Define
\be \label{E:Ydef}
Y=\operatorname{diag}\bigl(Y_1,Y_1^{-1}\bigr),
\qquad
Y_1
=
e^{\ii t\theta(\xi,k_c)}
d_0^{1/2}
e^{-\frac{\ii}{2}\arg r(k_c)} .
\ee

\begin{figure}
\centering
\begin{tikzpicture}[x=1.8cm,y=1.8cm,scale=0.6,transform shape]
  \coordinate (Kone) at (-0.625,0);
  \coordinate (Kc)   at ( 0.70,0);
  \coordinate (C)    at ( 0.05,0);
  \def\rad{2.25}

  \begin{scope}
    \clip (C) circle (\rad);

    \draw[flow=.78]
      (-6.10, 1.43)
        .. controls (-4.65, 1.48) and (-2.70, 1.08) .. (Kone);
    \draw[flow=.78]
      (-6.10,-1.43)
        .. controls (-4.65,-1.48) and (-2.70,-1.08) .. (Kone);

    \draw[axisflow=.55] (-6.10,0) -- (Kone);
    \draw[flow=.55] (Kone) -- (Kc);

    \draw[axisflow=.86] (Kc) -- (7.05,0);

    \draw[flow=.5]
      (Kc)-- (0.7, 2.14);
    \draw[flow=.5]
      (Kc)-- (0.7, -2.14);

    \draw[twoflow=.30 and .72] (4.40,-1.55) -- (4.40,1.55);
  \end{scope}

  \draw[boundarycircle] (C) circle (\rad);

  \node[contourlabel] at (-1.05, 0.78) {$\Sigma^{\eps}_1$};
  \node[contourlabel] at (-1.05,-0.78) {$\Sigma^{\eps}_2$};
  \node[contourlabel] at ( 1.10, 1.18) {$\Sigma^{\eps}_3$};
  \node[contourlabel] at ( 1.10,-1.18) {$\Sigma^{\eps}_4$};
  \node[contourlabel] at ( 0.06, 0.20) {$\Sigma^{\eps}_5$};

  \node[label,below=2pt] at (Kone) {$k_1$};
  \node[label,below=2pt,xshift=8pt] at (Kc) {$k_c$};
\end{tikzpicture}
\caption{The contour \(\Sigma^{\eps}=\bigcup_{j=1}^5 \Sigma_j^{\eps}\).}
\label{F:Sigma2-local}
\end{figure}
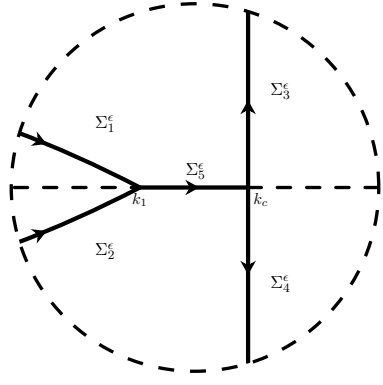

\noindent Set
\[
\widetilde M=M^{(4)}Y, \qquad k \in \cD_{\eps}.
\]
Then the jump matrix of \(\widetilde M\) is
\[
\widetilde V
=
Y^{-1}V^{(4)}Y.
\]
The purpose of this conjugation is to remove the constant phase and amplitude
factors in the local jumps.  Define  $\Sigma^{\eps}=\cup_{j=1}^5 \Sigma^{\eps}_j$, where
$\Sigma^{\eps}_j=\Sigma^{(4)}_j \cap \cD_{\eps}$; see Figure~\ref{F:Sigma2-local}.
Let $\widetilde{V}_j$ denote the restriction of $\widetilde{V}$ to $\Sigma^{\eps}_j$. For a fixed $z$,  as $t \to \infty$~(this leads to $k \to k_c$), we know that $\widetilde{V}_j(x,t,k)$ tends to the jump matrix $V^{X}_j(x,t,z)$ defined in~\eqref{E:Vjexp-minus} as $t\to\infty$.

Therefore, the local parametrix in the \(k\)-plane is defined by
\be \label{E:Mloc}
M^{\mathrm{loc}}(k)
=
Y M^X(z(k);y)Y^{-1},
\qquad
k\in \cD_\eps,
\ee
where $M^X(z;y)$ is the solution of the model RH problem~\ref{RHmP-minus}. We now show that \(M^{(4)}\) can be approximated by
\(M^{\rm loc}\) in \(\cD_\eps\).

\begin{lemma}\label{L:Nloc}
For each \((x,t)\), the function \(M^{\rm loc}(x,t,k)\) defined in~\eqref{E:Mloc}
is analytic and bounded for \(k \in\cD_\eps\setminus \Sigma^\eps\).  Across
\(\Sigma^\eps\), it satisfies
\[
M^{\rm loc}_+(x,t,k)=M^{\rm loc}_-(x,t,k) V^{\rm loc}(x,t,k).
\]
Moreover, for sufficiently large \(t\),
\begin{align}\label{E:estVloc}
\begin{cases}
\|V^{(4)}-V^{\rm loc}\|_{L^\infty(\Sigma^\eps)}
\le C t^{-1/3}\log t,\\[1mm]
\|V^{(4)}-V^{\rm loc}\|_{L^1(\Sigma^\eps)}
\le C t^{-2/3}\log t,
\end{cases}
\qquad (x,t)\in\cP_-.
\end{align}
Furthermore, as \(t\to+\infty\),
\begin{align}
\|(M^{\rm loc})^{-1}-I\|_{L^\infty(\partial\cD_\eps)}
&=\mathcal O(t^{-1/3}), \label{E:bzsy}\\
(M^{\rm loc})^{-1}(x,t,k)-I
&=
-\frac{
YM^X_1  Y^{-1}
}{
\left(\frac{8 \sqrt{6}}{9q_o}\right)^{1/3}t^{1/3}(k-k_c)
}
+\mathcal O(t^{-2/3}),
\quad k\in\partial\cD_\eps, \label{E:bzsy1}
\end{align}
where \( M^X_1\) is given by~\eqref{E:defMX1dd}.
\end{lemma}
\begin{proof}
By the definition of \(M^{\mathrm{loc}}\), its jump matrix
\(V^{\mathrm{loc}}\) satisfies
\[
V^{(4)}(k)-V^{\mathrm{loc}}(k)
=
Y
\bigl(\widetilde V(k)-V^X(z(k))\bigr)
Y^{-1}.
\]
Since \(Y\) and \(Y^{-1}\) are uniformly bounded in
\(\cD_{\eps}\), it is enough to estimate
\(\widetilde V(k)-V^X(z(k))\).

We illustrate the estimate on \(\Sigma_1^{\eps}\). On this contour,
\[
\widetilde V_1(k)-V_1^X(z(k))
=
\begin{pmatrix}
0 & g_1(k)\\
0 & 0
\end{pmatrix},
\]
where
\be \label{E:g1}
\begin{aligned}
g_1(k)
={}&
\frac{\overline{r}(k)}{1+r(k)\overline{r}(k)}
e^{\mathrm{i}\arg r(k_c)}
e^{2\mathrm{i}t(\theta(\xi,k)-\theta(\xi,k_c))}
z^{2\mathrm{i}\nu}d_1(k)
-
\frac{r_0}{1+r_0^2}
e^{\mathrm{i}(yz+z^3)}
z^{2\mathrm{i}\nu}.
\end{aligned}
\ee
Therefore it remains to estimate \(g_1(k)\). We write
\[
\begin{aligned}
|g_1(k)|
\le{}&
\left|
\left(
\frac{\overline{r}(k)}{1+r(k)\overline{r}(k)}
e^{\mathrm{i}\arg r(k_c)}
-
\frac{r_0}{1+r_0^2}
\right)
e^{2\mathrm{i}t(\theta(\xi,k)-\theta(\xi,k_c))}
z^{2\mathrm{i}\nu}d_1(k)
\right|
\\
&+
\left|
\frac{r_0}{1+r_0^2}
z^{2\mathrm{i}\nu}d_1(k)
\left(
e^{2\mathrm{i}t(\theta(\xi,k)-\theta(\xi,k_c))}
-
e^{\mathrm{i}(yz+z^3)}
\right)
\right|
\\
&+
\left|
\frac{r_0}{1+r_0^2}
z^{2\mathrm{i}\nu}
e^{\mathrm{i}(yz+z^3)}
\bigl(d_1(k)-1\bigr)
\right|.
\end{aligned}
\]
For \( k \in \Sigma_1^\eps\),  it is straightforward to verify that
\[ \left| e^{2\ii t\theta(\xi,k)} \right| \le C e^{-ct|k-k_c|^3}, \qquad \left| e^{\ii(yz+z^3)} \right| \le C e^{-ct|k-k_c|^3}. \] 
Since \(\mathrm{Re}(\ii t\theta(\xi,k_c))=0\), and using the elementary inequality \[ |e^{\omega_1}-e^{\omega_2}| \le \left(|e^{\omega_1}|+|e^{\omega_2}|\right) |\omega_1-\omega_2|, \qquad \omega_1,\omega_2\in\C, \] we obtain
\[ \left| e^{2\ii t(\theta(\xi,k)-\theta(\xi,k_c))} - e^{\ii(yz+z^3)} \right| \le Ct|S(\xi,k)|e^{-ct|k-k_c|^3}. \]
Therefore,
\[
|g_1(k)|
\le
C|k-k_c|e^{-ct|k-k_c|^3}
+
Ct|S(\xi,k)|e^{-ct|k-k_c|^3}
+
C|d_1(k)-1|e^{-ct|k-k_c|^3}.
\]
Moreover, by the definition of \(S(\xi,k)\)~\eqref{E:Sdef},
\[
|S(\xi,k)|
\le
C|\xi-\xi_c||k-k_c|^2
+
C|\xi-\xi_c||k-k_c|^3
+
C|k-k_c|^4.
\]
On the other hand, the local logarithmic representation gives
\[
|d_1(k)-1|
\le
C\left(1+\left|\log|k-k_c|\right|\right)|k-k_c|.
\]
Hence, putting \(u=|k-k_c|\), and using the  scaling
\(
|\xi-\xi_c|\le Ct^{-2/3}
\),
we get
\[
\begin{aligned}
|g_1(k)|
\le{}&
C\sup_{u>0} u e^{-ctu^3}
+
Ct|\xi-\xi_c|\sup_{u>0} u^2 e^{-ctu^3} + Ct|\xi-\xi_c|\sup_{u>0} u^3 e^{-ctu^3}\\
&+
Ct\sup_{u>0} u^4 e^{-ctu^3}
+
C\sup_{u>0}
\left(1+|\log u|\right)u e^{-ctu^3}.
\end{aligned}
\]
Consequently,
\[
|g_1(k)|
\le
C t^{-1/3}\log t .
\]
Similarly, by direct integration along \(\Sigma_1^{\eps}\), we obtain
\[
\int_{\Sigma_1^{\eps}}|g_1(s)|\,|ds|
\le
C t^{-2/3}\log t .
\]
The estimates on the other components of the local contour are obtained in the
same way.
\end{proof}

\vskip 2mm
\noindent\textbf{Outer parametrix.} 
In the present setting, the outer parametrix is the solution of the following
RH problem:
\[
\begin{cases}
M^{\mathrm{out}}_+(k)
=
M^{\mathrm{out}}_-(k)V_B(k),
&
k\in B,
\\[0.8ex]
M^{\mathrm{out}}(k)
=
\left(I+O(k^{-1})\right)e^{\mathrm{i}g_\infty\sigma_3},
&
k\to\infty.
\end{cases}
\]
This model problem has been explicitly solved in Refs.
\cite{BM2017,BLM2021}. More precisely, its solution is given by
\be \label{E:Mout}
M^{\mathrm{out}}(k)
=
\frac12 e^{\mathrm{i}g_\infty\sigma_3}
\begin{pmatrix}
\Lambda(k)+\Lambda^{-1}(k)
&
-\dfrac{q_-}{q_o}\left(\Lambda(k)-\Lambda^{-1}(k)\right)
\\[1.2ex]
-\dfrac{q_o}{q_-}\left(\Lambda(k)-\Lambda^{-1}(k)\right)
&
\Lambda(k)+\Lambda^{-1}(k)
\end{pmatrix},
\ee
where
\[
\Lambda(k)
=
\left(
\frac{k-\mathrm{i}q_o}{k+\mathrm{i}q_o}
\right)^{1/4}.
\]
Here the fourth root is chosen so that
\[
\Lambda(k)\to 1,
\qquad k\to\infty.
\]
Notice that
\[
\lim_{k\to\infty}
k\left(\Lambda(k)-\Lambda^{-1}(k)\right)
=
-\mathrm{i}q_o.
\]
Therefore,
\be \label{E:gjjx}
\begin{aligned}
-2\mathrm{i}
\lim_{k\to\infty}
\left[
k\left(M^{\mathrm{out}}(x,t,k)\right)_{12}
\right]
e^{\mathrm{i}g_\infty}
&=
-2\mathrm{i}
\lim_{k\to\infty}
\left[
-\frac{q_-}{2q_o}
k\left(\Lambda(k)-\Lambda^{-1}(k)\right)
e^{\mathrm{i}g_\infty}
\right]
e^{\mathrm{i}g_\infty}
\\
&=
q_-e^{2\mathrm{i}g_\infty}.
\end{aligned}
\ee
This limiting relation will be used later in the reconstruction formula.

\subsection{The small norm RH problem} \label{su:smn}

We define the final transformation to obtain a small-norm RH problem as follows:
\[
E(x,t,k)
=
\begin{cases}
M^{(4)}(x,t,k)\bigl(M^{\mathrm{out}}(x,t,k)\bigr)^{-1},
&
k\in \mathbb C\setminus \overline{\cD_\eps},
\\[1.2ex]
M^{(4)}(x,t,k)
\bigl(M^{\mathrm{loc}}(x,t,k)\bigr)^{-1}
\bigl(M^{\mathrm{out}}(x,t,k)\bigr)^{-1},
&
k\in \cD_\eps .
\end{cases}
\]
The jump contour of \(E\) is
\[
\Sigma^E=\Sigma^{(4)}\cup \partial \cD_\eps,
\]
where \(\partial \cD_\eps\) is oriented counterclockwise.
The corresponding jump matrix \(V^E\) is given by
\be 
V^E(k)
=
\begin{cases}
M^{\mathrm{out}}(k)
\bigl(M^{\mathrm{loc}}(k)\bigr)^{-1}
\bigl(M^{\mathrm{out}}(k)\bigr)^{-1},
&
k\in \partial \cD_\eps,
\\[1.4ex]
M^{\mathrm{out}}(k)
V^{(4)}(k)
\bigl(M^{\mathrm{out}}(k)\bigr)^{-1},
&
k\in \Sigma^{(4)}
\setminus\bigl(\overline{ \cD_\eps}\cup B\bigr),
\\[1.4ex]
I,
&
k\in B,
\\[1.4ex]
M^{\mathrm{out}}(k)
\left[
M_-^{\mathrm{loc}}(k)
V^{(4)}(k)
\bigl(M_+^{\mathrm{loc}}(k)\bigr)^{-1}
\right]
\bigl(M^{\mathrm{out}}(k)\bigr)^{-1},
&
k\in \Sigma^{\eps}.
\end{cases}  \label{E:VEdef}
\ee
Let
\[
w^E(k):=V^E(k)-I.
\]
The following estimates show that the error problem is a small-norm
RH problem.

\begin{lemma} \label{E:mxy}
As \(t\to\infty\), the following estimates hold uniformly for
\((x,t)\in\mathcal P_-\):
\begin{subequations}\label{E:wE-estimates}
\begin{align}
&\left\|w^E\right\|_{L^1\cap L^\infty
\left(\Sigma^E\setminus \overline{\cD_\eps}\right)}
\le
C e^{-ct},
\label{E:wE-estimates-a}
\\
&\left\|w^E\right\|_{L^1\cap L^\infty(\partial \cD_\eps)}
\le
C t^{-1/3},
\label{E:wE-estimates-b}
\\
&\left\|w^E\right\|_{L^1\left(\Sigma^{\eps}\right)}
\le
C t^{-2/3}\log t,
\label{E:wE-estimates-c}
\\
&\left\|w^E\right\|_{L^\infty\left(\Sigma^{\eps}\right)}
\le
C t^{-1/3}\log t.
\label{E:wE-estimates-d}
\end{align}
\end{subequations}

\end{lemma}
\begin{proof} For \(k\in \Sigma^{(4)}\setminus \bigl(\overline{\cD_\eps}\cup B\bigr)\), the sign distribution of \(\Re(\ii\theta)\) implies that the jump matrix \(V^{(4)}(k)\) converges exponentially fast to the identity matrix \(I\). Moreover, on the branch cut \(B\), the error jump is trivial, namely \(V^E=I\). These facts give \eqref{E:wE-estimates-a}. Next, the estimate on \(\partial\cD_\eps\) follows directly from the matching condition~\eqref{E:bzsy1} and the definition of the error jump matrix in~\eqref{E:VEdef}. This proves~\eqref{E:wE-estimates-b}. It remains to estimate \(w^E\) on the local contour \(\Sigma^\eps\). For \(k\in\Sigma^\eps\), we have \[ w^E(k) = M^{\mathrm{out}}(k) \left[ M_-^{\mathrm{loc}}(k) \left( V^{(4)}(k)-V^{\mathrm{loc}}(k) \right) \bigl(M_+^{\mathrm{loc}}(k)\bigr)^{-1} \right] \bigl(M^{\mathrm{out}}(k)\bigr)^{-1}. \] Since \(M^{\mathrm{out}}\), \(M^{\mathrm{loc}}\), and their inverses are uniformly bounded in \(\cD_\eps\), the estimates \eqref{E:wE-estimates-c} and~\eqref{E:wE-estimates-d} follow immediately from~\eqref{E:estVloc}. \end{proof}

Consequently,
\[
\left\|w^E\right\|_{L^1(\Sigma^E)}
\le
C t^{-1/3},
\qquad
\left\|w^E\right\|_{L^\infty(\Sigma^E)}
\le
C t^{-1/3}\log t.
\]
In particular, by interpolation,
\(
\left\|w^E\right\|_{L^2(\Sigma^E)}
\le
C t^{-1/3}(\log t)^{1/2}.
\)
Define the Cauchy operator
\[
\cC_{w^E}(h):=\cC_-\bigl(hw^E\bigr).
\]
Since
$
\|\cC_{w^E}\|_{B(L^2(\Sigma^E))}\to 0$, $t\to\infty
$,
the operator \(I-\cC_{w^E}\) is invertible for  sufficiently large \(t\).
Thus we may define
\[
\mu^E
=
I+(I-\cC_{w^E})^{-1} \cC_{w^E}I
\in I+L^2(\Sigma^E).
\]
We shall estimate \(\mu^E-I\). Indeed,
\be \label{E:mujh}
\begin{aligned}
\|\mu^E-I\|_{L^2(\Sigma^E)}
&\le
\sum_{j=0}^{\infty}
\| \cC_{w^E}\|_{B(L^2(\Sigma^E))}^j
\|\cC_{w^E}I\|_{L^2(\Sigma^E)}
\\
&\le
\frac{
\|\cC_-\|_{B(L^2(\Sigma^E))}
\|w^E\|_{L^2(\Sigma^E)}
}{
1-
\|\cC_-\|_{B(L^2(\Sigma^E))}
\|w^E\|_{L^\infty(\Sigma^E)}
}
 \le
C t^{-1/3}(\log t)^{1/2}.
\end{aligned}
\ee

By the standard theory for small-norm RH problems~\cite{Lenells2018,Lenells2017}, the error
function \(E(x,t,k)\) can be represented as
\[
E(x,t,k)
=
I+\frac{1}{2\pi\mathrm{i}}
\int_{\Sigma^E}
\frac{\mu^E(x,t,s)w^E(x,t,s)}{s-k}\,ds,
\qquad
k\in\mathbb C\setminus\Sigma^E.
\]
Therefore, the coefficient of \(k^{-1}\) in the expansion of \(E\) at infinity
is
\[
E^{(1)}(x,t)
:=
\lim_{k\to\infty}k\bigl(E(x,t,k)-I\bigr)=-\frac{1}{2\pi\mathrm{i}}
\int_{\Sigma^E}
\mu^E(x,t,s)w^E(x,t,s)\,ds.
\]

\begin{lemma}
As \(t\to\infty\),
\be \label{E:E1asi}
E^{(1)}(x,t)
=
-\frac{1}{2\pi\mathrm{i}}
\int_{\partial \cD_\eps}
w^E(x,t,s)\,ds
+
\mathcal O\left(t^{-2/3}\log t\right).
\ee
\end{lemma}

\begin{proof}
From the Cauchy representation above, we have
\[
E^{(1)}(x,t)
=
-\frac{1}{2\pi\mathrm{i}}
\int_{\partial \cD_\eps}
w^E(x,t,s)\,ds
+
F_1(x,t)+F_2(x,t),
\]
where
\[
F_1(x,t)
=
-\frac{1}{2\pi\mathrm{i}}
\int_{\Sigma^E\setminus\partial \cD_\eps}
w^E(x,t,s)\,ds,
\]
and
\[
F_2(x,t)
=
-\frac{1}{2\pi\mathrm{i}}
\int_{\Sigma^E}
\bigl(\mu^E(x,t,s)-I\bigr)w^E(x,t,s)\,ds.
\]
Using Lemma~\ref{E:mxy} and the bound for \(\mu^E-I\)~\eqref{E:mujh}, we obtain
\[
F_1(x,t)+F_2(x,t)
=
\mathcal O\left(t^{-2/3}\log t\right).
\]
This proves the lemma.
\end{proof}

We next compute the contribution from \(\partial \cD_\eps \). Recall that
for \(k\in\partial \cD_\eps\),
\[
w^E(x,t,k)
=
M^{\mathrm{out}}(x,t,k)
\left[
\bigl(M^{\mathrm{loc}}(x,t,k)\bigr)^{-1}-I
\right]
\bigl(M^{\mathrm{out}}(x,t,k)\bigr)^{-1}.
\]
Therefore, by~\eqref{E:bzsy1} and  the residue theorem,
\be \label{E:wedgj}
\begin{aligned}
-\frac{1}{2\pi\mathrm{i}}
\int_{\partial \cD_\eps}
w^E(x,t,s)\,ds
&=
\frac{
M^{\mathrm{out}}(x,t,k_c)
Y M_1^X(y) Y^{-1}
\bigl(M^{\mathrm{out}}(x,t,k_c)\bigr)^{-1}
}{
\left(\tfrac{8\sqrt{6}}{9q_o}\right)^{1/3}
t^{1/3}
}
\\
&\quad
+
\mathcal O\left(t^{-2/3}\right),
\qquad
t\to\infty.
\end{aligned}
\ee
For convenience, write
\be \label{E:defbetaj}
M^{\mathrm{out}}(x,t,k_c)
=:
\begin{pmatrix}
\beta_{11} & \beta_{12}
\\
\beta_{21} & \beta_{22}
\end{pmatrix}.
\ee
Using formula~\eqref{E:Ydef},~\eqref{E:E1asi},~\eqref{E:wedgj},~\eqref{E:defMX1dd}, \eqref{E:defbetaj} and $\det M^{\rm out}=1$, a long but direct calculation gives
\be \label{E:E12asyp}
E^{(1)}_{12}(x,t)
=
-
\frac{
\alpha_1\mathcal V(y)
+
\alpha_2\mathcal V^*(y)
+
\alpha_3\widehat{\mathcal V} (y)
}{
\left(\tfrac{8\sqrt{6}}{9q_o}\right)^{1/3}t^{1/3}
}
+
\mathcal O\left(t^{-2/3} \log t\right), \qquad t \to \infty,
\ee
where
\be \label{E:defalpha}
\alpha_1=\beta_{12}^2 Y_1^{-2},
\qquad
\alpha_2=\beta_{11}^2 Y_1^{2},
\qquad
\alpha_3=2\beta_{11}\beta_{12}.
\ee
Moreover, it is easy to see that
\(\alpha_1(x,t)\), \(\alpha_2(x,t)\), and \(\alpha_3(x,t)\) are uniformly
bounded quantities.

\subsection{
Proof of Theorem~\ref{Th:main} for \texorpdfstring{\((x,t)\in\cP_-\)}{(x,t) in P-}
}\label{S:pthmain}

We now use the reconstruction formula
\[
q(x,t)
=
-2\mathrm{i}\lim_{k\to\infty} kM_{12}(x,t,k)
\]
to compute the long-time asymptotics of \(q(x,t)\). Recalling all the invertible transformations introduced above,  we have
\[
E(x,t,k)
=
M(x,t,k)\delta(k)^{-\sigma_3}D(k)
e^{\mathrm{i}g(k)\sigma_3}
\bigl(M^{\mathrm{out}}(x,t,k)\bigr)^{-1}.
\]
Therefore,
\[
M(x,t,k)
=
E(x,t,k)M^{\mathrm{out}}(x,t,k)
e^{-\mathrm{i}g(k)\sigma_3}
D^{-1}(k)\delta(k)^{\sigma_3}.
\]
Let
\[
M_{12}^{(1)}(x,t)
:=
\lim_{k\to\infty}kM_{12}(x,t,k),
\qquad
\bigl(M^{\mathrm{out}}\bigr)_{12}^{(1)}(x,t)
:=
\lim_{k\to\infty}kM_{12}^{\mathrm{out}}(x,t,k).
\]
Then the reconstruction formula gives
\[
q(x,t)
=
-2\mathrm{i}M_{12}^{(1)}(x,t)=-2\ii E^{(1)}_{12}(x,t)-2\ii \bigl(M^{\mathrm{out}}\bigr)_{12}^{(1)}(x,t)e^{\ii g_{\infty}}.
\]
Using~\eqref{E:E12asyp} and~\eqref{E:gjjx}, we arrive at
\[
q(x,t)
=
q_-e^{2\mathrm{i}g_\infty}
+
\frac{
2\mathrm{i}\left[
\alpha_1\mathcal V(y)
+
\alpha_2\mathcal V^*(y)
+
\alpha_3\widehat {\mathcal V}(y)
\right]
}{
\left(\tfrac{8\sqrt{6}}{9q_o}\right)^{1/3}t^{1/3}
}
+
\mathcal O\left(t^{-2/3}\log t\right).
\]
This is precisely the asymptotic formula~\eqref{E:asygs}.


\section{Long--time asymptotics in the transition region \texorpdfstring{$\cP_+$}{P+}}\label{S:dzcpplus}
The boundary transition analysis from the modulated elliptic-wave side differs
substantially from the analysis in the interior of the elliptic-wave region.
In~\cite[Section~5]{BM2017}, a \(G\)-function transformation is introduced to
control the nondecaying jump.  In the present transition regime, however, the
nondecaying contribution is confined to a short portion of the contour of
length \(\mathcal O(t^{-1/3})\), and can therefore be absorbed into the local
model problem near \(k_c\).

We explain this point in more detail.   First, let us continue the contour directly from
\(k_c\) upward and downward, and assume that the initial parts of these
continuations are vertical.  Denote by \(k_3\) and \(k_4\) the corresponding
intersection points with the level curve
\(
        \mathrm{Re}(\ii\theta(\xi,k))=0
\)
on the upper and lower sides, respectively.  Then the jumps on the short
segments connecting \(k_4\) to \(k_c\) and \(k_c\) to \(k_3\) are exponentially
growing as \(t\to\infty\).  However, the length of this growing part is only of
order \(\mathcal O(t^{-1/3})\).  Indeed, by using the Taylor expansion of
\(\theta(\xi,k)\) near \(k_c\), a direct calculation gives
\[
|k_3-k_c|^2
=
\frac{3q_o}{4\sqrt{2}}(\xi-\xi_c)
\left(1+\mathcal O(\sqrt{\xi-\xi_c})\right).
\]
Since \(\xi-\xi_c=\mathcal O(t^{-2/3})\) in the transition scaling, it follows
that
\[
        |k_3-k_c|=\mathcal O(t^{-1/3}),
\]
and the same estimate holds for \(k_4\).  Thus the exponentially growing part
of the jump is entirely contained in the local neighborhood of \(k_c\) and can
be incorporated into the local model problem.  Moreover, we will show below
that the local parametrix gives a sufficiently accurate approximation
throughout this neighborhood.

Consequently, no additional \(G\)-function transformation is required.  The
analysis therefore follows the same general scheme as in the region
\(\cP_-\).  We only describe the necessary modifications and omit the details
that are identical to those in Section~\ref{S:dza}.

\medskip
\noindent\textbf{Opening of lenses.}
The first transformation is slightly different from the one used for
\((x,t)\in\cP_-\).  According to the sign distribution of
\(\mathrm{Re}(\ii\theta)\), see Figure~\ref{sign-structure-f}, we should now open the lens along
\((-\infty,k_c)\).  The second, third, and fourth transformations are the
same as those in the case \((x,t)\in\cP_-\), and hence we do not repeat their
details here.  After this sequence of transformations, we arrive at the
function \(M^{(4)}\).  The corresponding jump contour is denoted by
\(\Sigma^{(4)}\), see Figure~\ref{fig:deformed-contour}, and its jump matrices
are denoted by \(V^{(4)}\), including the jump \(V_B\) on the branch cut.  The
matrices \(\{V_j^{(4)}\}_{j=1}^4\) and \(V_B\) have the same form as in
Section~\ref{S:dza}.  In particular, except on the branch cut and in a small
neighborhood of \(k_c\), all jumps are exponentially close to the identity
matrix as \(t\to\infty\).
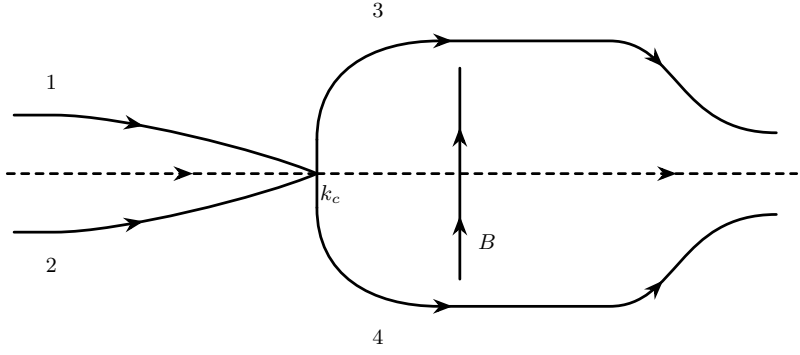
\begin{figure}
\centering
\scalebox{0.9}{%
\begin{tikzpicture}[x=1cm,y=1cm]

\tikzset{
  contour/.style={draw=black, line width=1.15pt, line cap=round, line join=round},
  dashedaxis/.style={draw=black, dashed, line width=1.15pt, line cap=round, line join=round},
  flow/.style={
        draw=black,
        line width=1.15pt,
        line cap=round,
        line join=round,
        postaction={decorate},
        decoration={
            markings,
            mark=at position #1 with {\arrow{Stealth}}
        }
    },
twoflow/.style args={#1 and #2}{
        draw=black,
        line width=1.15pt,
        line cap=round,
        line join=round,
        postaction={decorate},
        decoration={
            markings,
            mark=at position #1 with {\arrow{Stealth}},
            mark=at position #2 with {\arrow{Stealth}}
        }
    },
  axisflow/.style={
        draw=black,
        dashed,
        line width=1.15pt,
        line cap=round,
        line join=round,
        postaction={decorate},
        decoration={
            markings,
            mark=at position #1 with {\arrow{Stealth}}
        }
    },
  every node/.style={font=\footnotesize, inner sep=1pt}
}

\coordinate (Btop)   at (0,1.72);
\coordinate (Bbot)   at (0,-1.72);
\coordinate (Kzero)  at (-2.10,0);
\coordinate (Gup)    at (-2.10,0.68);
\coordinate (Gdown)  at (-2.10,-0.68);

\draw[dashedaxis,axisflow=0.60] (-6.65,0) -- (Kzero);
\draw[dashedaxis,axisflow=0.75] (Kzero) -- (4.95,0);

\draw[twoflow=.30 and .72] (0,-1.55) -- (0,1.55);

\draw[contour,flow=0.42]
  (-6.55, 0.86) -- (-5.95,0.86)
  .. controls (-5.00,0.86) and (-2.90,0.34) .. (Kzero);

\draw[contour,flow=0.42]
  (-6.55,-0.86) -- (-5.95,-0.86)
  .. controls (-5.00,-0.86) and (-2.90,-0.34) .. (Kzero);

  \draw[draw=black,line width=1.15pt,line cap=round] (Kzero)--(-2.1,0.50);
  \draw[draw=black,line width=1.15pt,line cap=round] (Kzero)--(-2.1,-0.50);
  \draw[twoflow=.35 and .75]
    (-2.1,0.50)
      .. controls (-2.1, 1.5) and (-1.3, 1.95) .. (-0.3, 1.95)
      .. controls (0.7, 1.95) .. (2.2, 1.95)
      .. controls (3.2, 1.95) and (3.2, 0.6) .. (4.65, 0.6);
  \draw[twoflow=.35 and .75]
    (-2.1,-0.50)
      .. controls (-2.1, -1.5) and (-1.3, -1.95) .. (-0.3, -1.95)
      .. controls (0.7, -1.95) .. (2.2, -1.95)
      .. controls (3.2, -1.95) and (3.2, -0.6) .. (4.65, -0.6);

\node at (-6.00, 1.35) {$1$};
\node at (-6.00,-1.35) {$2$};
\node at (0.40,-1.00) {$B$};

\node at (-1.20, 2.40) {$3$};
\node at (-1.20,-2.40) {$4$};
\node at (-1.90,-0.30) {$k_c$};

\end{tikzpicture}%
}
\caption{The jump contour for \(M^{(4)}\).}
\label{fig:deformed-contour}
\end{figure}

\medskip
\noindent\textbf{Local parametrix.}
The construction of the local parametrix is also analogous to that in the case
\((x,t)\in\cP_-\).  Let \(\cD_{\eps}\) be the open disk of radius \(\eps\)
centered at \(k_c\).  Define
\[
   \Sigma^{\eps}=\bigcup_{j=1}^4 \Sigma_j^{\eps},
   \qquad
   \Sigma_j^{\eps}=\Sigma_j^{(4)}\cap\cD_{\eps},
\]
see Figure~\ref{fig:local-kzero}.

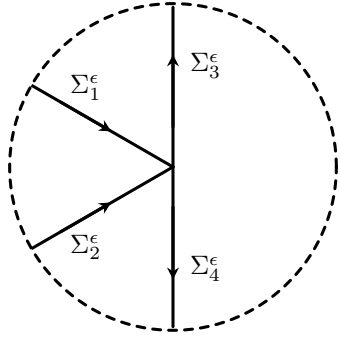
\begin{figure}
\centering
\begin{tikzpicture}[x=2.4cm,y=2.4cm]

\tikzset{
  contour/.style={draw=black,line width=1.15pt,line cap=round,line join=round},
  dashedcircle/.style={draw=black,dashed,line width=1.15pt,line cap=round,line join=round},
  dirarrow/.style={draw=black,line width=1.15pt,-{Stealth[length=4.5pt,width=4.5pt]},line cap=round,line join=round}
}

\coordinate (Kzero) at (0,0);
\def\r{0.90}

\coordinate (Lup) at (-0.7794, 0.45);
\coordinate (Ldn) at (-0.7794,-0.45);

\begin{scope}
  \clip (Kzero) circle (\r);

  \draw[contour] (Lup) -- (Kzero);
  \draw[contour] (Ldn) -- (Kzero);

  \draw[contour] (Kzero) -- (0, 0.88);
  \draw[contour] (Kzero) -- (0,-0.88);

  \draw[dirarrow] (-0.62, 0.358) -- (-0.34, 0.196);
  \draw[dirarrow] (-0.62,-0.358) -- (-0.34,-0.196);

  \draw[dirarrow] (0,0.22) -- (0,0.62);
  \draw[dirarrow] (0,-0.22) -- (0,-0.62);
\end{scope}

\draw[dashedcircle] (Kzero) circle (\r);

\node at (-0.48, 0.44) {{\small $\Sigma^{\eps}_1$}};
\node at (-0.48,-0.44) {{\small $\Sigma^{\eps}_2$}};
\node at (0.18, 0.56) {{\small $\Sigma^{\eps}_3$}};
\node at (0.18,-0.56) {{\small $\Sigma^{\eps}_4$}};

\end{tikzpicture}
\caption{The local contour $\Sigma^{\eps}=\cup_{j=1}^4 \Sigma^{\eps}_j$.}
\label{fig:local-kzero}
\end{figure}

The local variables \(y\) and \(z\) are defined as in~\eqref{E:yz}, while
\(d_0\), \(d_1\), and the prefactor matrix \(Y\) are those introduced in
Section~\ref{S:dza}.  We define
\be \label{E:xiyMlo}
M^{\mathrm{loc}}(x,t,k)
=
Y M^X(z;y)Y^{-1}, 
\qquad k\in\cD_{\eps}, \qquad (x,t) \in \cP_+.
\ee
Here \(M^X\) is the solution of the model RH problem~\ref{RHmP-minus}.  Although this definition of \(M^{\mathrm{loc}}\) is not
identical to the one used for \((x,t)\in\cP_-\), we use the same notation; the
meaning will be clear from the context.

We now prove that the local parametrix \(M^{\mathrm{loc}}\) approximates
\(M^{(4)}\) with sufficient accuracy inside \(\cD_{\eps}\).  The following
exponential estimates are the key input for the matching argument.

\begin{lemma}\label{L:exp-est-Pplus}
Let \(\eps>0\) be sufficiently small.  Then there exist constants \(C,c>0\)
such that, for all sufficiently large \(t\) and uniformly for
\((x,t)\in\cP_+\),
\begin{subequations}\label{E:exp-est}
\begin{align}
\left|e^{2\ii t\theta(\xi,k)}\right|
&\le
C e^{-ct|k-k_c|^3},
\qquad
k\in \Sigma_1^\eps\cup\Sigma_4^\eps,
\label{E:exp-est-a}
\\
\left|e^{-2\ii t\theta(\xi,k)}\right|
&\le
C e^{-ct|k-k_c|^3},
\qquad
k\in \Sigma_2^\eps\cup\Sigma_3^\eps.
\label{E:exp-est-b}
\end{align}
\end{subequations}
\end{lemma}

\begin{proof}
We give the proof on a representative part of the local contour; the remaining
estimates are obtained in the same way.  Parametrize $\Sigma^{\eps}_3$ by
\[
        k=k_c+\ii s,\qquad 0\le s\le\eps.
\]
Near \(k_c\), the phase function admits the expansion
\[
\theta(\xi,k)
=
\theta(\xi,k_c)
+\frac{\xi-\xi_c}{\sqrt{3}}(k-k_c)
+\frac{4\sqrt{6}}{9q_o}(k-k_c)^3
+R(\xi,k),
\]
where \(\theta(\xi,k_c)\in\R\), and the remainder satisfies
\[
|R(\xi,k)|
\le
C_1\left(
|\xi-\xi_c|\,|k-k_c|^2+|k-k_c|^4
\right).
\]
For convenience, set
\[
        A=\frac{4\sqrt{6}}{9q_o},
        \qquad
        B=\frac{1}{\sqrt{3}},
        \qquad
        \Delta=\xi-\xi_c .
\]
Since \(\theta(\xi,k_c)\in\R\), the constant term in the local expansion does
not contribute to the imaginary part.  Substituting \(k=k_c+\ii s\) into the
above expansion and taking imaginary parts gives
\[
\operatorname{Im}\theta(\xi,k)
\le
B\Delta s
-As^3
+C_1\left(\Delta s^2+s^4\right).
\]

We now control the positive terms.  By Young's inequality, for
\(\Delta>0\) and \(s>0\),
\[
\Delta s
\le
\frac{A}{4B}s^3+C_2\Delta^{3/2},
\qquad
\Delta s^2
\le
\frac{A}{8C_1}s^3+C_3\Delta^3 .
\]
Hence, for some constant \(C_4>0\),
\[
B\Delta s+C_1\Delta s^2
\le
\frac{3A}{8}s^3
+C_4\left(\Delta^{3/2}+\Delta^3\right).
\]
Moreover, after decreasing \(\eps\) if necessary, we have
\[
C_1s^4\le \frac{A}{8}s^3,
\qquad 0<s\le \eps.
\]
It follows that
\[
\operatorname{Im}\theta(\xi,k)
\le
-\frac{A}{2}s^3
+C_4\left(\Delta^{3/2}+\Delta^3\right).
\]
In the transition region \(\cP_+\), we have \(0<\Delta\le C t^{-2/3}\).  Thus
there exists a constant \(C_5>0\) such that
\[
 t\,\operatorname{Im}\theta(\xi,k)
\le
-\frac{A}{2}ts^3+C_5 .
\]
Consequently, for suitable constants \(C,c>0\),
\[
\left|e^{-2\ii t\theta(\xi,k)}\right|
\le
C e^{-cts^3}
=
C e^{-ct|k-k_c|^3}.
\]
This proves the desired estimate on the chosen part of the contour.  The other
estimates follow in the same way.
\end{proof}

\begin{remark}\label{R:mpedphseest}
The proof of Lemma~\ref{L:exp-est-Pplus} also gives the corresponding
estimates for the model exponential factors:
\begin{subequations}\label{E:model-exp-est}
\begin{align}
\left|
\exp\{-\ii yz-\ii z^3\}
\right|
&\le
C e^{-ct|k-k_c|^3},
\qquad
k\in \Sigma_2^\eps\cup\Sigma_3^\eps,
\label{E:model-exp-est-b}
\\
\left|
\exp\{\ii yz+\ii z^3\}
\right|
&\le
C e^{-ct|k-k_c|^3},
\qquad
k\in \Sigma_1^\eps\cup\Sigma_4^\eps.
\label{E:model-exp-est-a}
\end{align}
\end{subequations}
\end{remark}

We are now ready to compare \(M^{\mathrm{loc}}\) with \(M^{(4)}\) in the disk
\(\cD_{\eps}\).

\begin{lemma}\label{L:Nloc-Pplus}
For each \((x,t)\in\cP_+\), the function \(M^{\mathrm{loc}}(x,t,k)\) defined
in~\eqref{E:xiyMlo} is analytic and bounded for
\(k\in\cD_{\eps}\setminus\Sigma^{\eps}\).  Across \(\Sigma^{\eps}\), it
satisfies
\[
M^{\mathrm{loc}}_+(x,t,k)
=
M^{\mathrm{loc}}_-(x,t,k)V^{\mathrm{loc}}(x,t,k).
\]
Moreover, for sufficiently large \(t\),
\begin{equation}\label{E:estVloc-1}
\begin{cases}
\|V^{(4)}-V^{\mathrm{loc}}\|_{L^\infty(\Sigma^\eps)}
\le C t^{-1/3}\log t,\\[1mm]
\|V^{(4)}-V^{\mathrm{loc}}\|_{L^1(\Sigma^\eps)}
\le C t^{-2/3}\log t,
\end{cases}
\qquad (x,t)\in\cP_+.
\end{equation}
Furthermore, as \(t\to+\infty\),
\begin{align}
\|(M^{\mathrm{loc}})^{-1}-I\|_{L^\infty(\partial\cD_\eps)}
&=\mathcal O(t^{-1/3}), \label{E:bzsy-1}
\\
(M^{\mathrm{loc}})^{-1}(x,t,k)-I
&=
-\frac{Y M^X_1Y^{-1}}
{\left(\frac{8\sqrt{6}}{9q_o}\right)^{1/3}t^{1/3}(k-k_c)}
+\mathcal O(t^{-2/3}),
\qquad k\in\partial\cD_\eps, \label{E:bzsy1-1}
\end{align}
where \(M^X_1\) is given by~\eqref{E:defMX1dd}.
\end{lemma}

\begin{proof}
The proof follows the same lines as the corresponding local-parametrix estimate
for \(\cP_-\).  We only prove the jump estimate~\eqref{E:estVloc-1} on
\(\Sigma_3^\eps\); the estimates on the other parts of the contour are
obtained similarly.

On \(\Sigma_3^\eps\), the local jump satisfies
\(V_3^{\mathrm{loc}}=YV_3^X Y^{-1}\).  Therefore,
\[
V_3^{(4)}-V_3^{\mathrm{loc}}
=
Y\bigl(\widetilde V_3-V_3^X\bigr)Y^{-1},
\qquad
\widetilde V_3=Y^{-1}V_3^{(4)}Y.
\]
For \(k\in\Sigma_3^\eps\), a direct calculation gives
\[
\widetilde V_3-V_3^X
=
\begin{pmatrix}
0&0\\
\widetilde g_1(k)&0
\end{pmatrix},
\]
where
\[
\begin{aligned}
\widetilde g_1(k)
&=
r(k)e^{-\ii\arg r(k_c)}
e^{-2\ii t(\theta(\xi,k)-\theta(\xi,k_c))}
z^{-2\ii\nu}d_1(k)
-r_0e^{-\ii(yz+z^3)}z^{-2\ii\nu}.
\end{aligned}
\]
Here \(r_0=|r(k_c)|\).  Hence
\[
\begin{aligned}
|\widetilde g_1(k)|
&\le
\left|
\bigl(r(k)e^{-\ii\arg r(k_c)}-r_0\bigr)
e^{-2\ii t(\theta(\xi,k)-\theta(\xi,k_c))}
z^{-2\ii\nu}d_1(k)
\right|
\\
&\quad+
\left|
r_0d_1(k)z^{-2\ii\nu}
\left(
e^{-2\ii t(\theta(\xi,k)-\theta(\xi,k_c))}
-e^{-\ii(yz+z^3)}
\right)
\right|
\\
&\quad+
\left|
r_0z^{-2\ii\nu}e^{-\ii(yz+z^3)}
\bigl(d_1(k)-1\bigr)
\right|.
\end{aligned}
\]
By Lemma~\ref{L:exp-est-Pplus}, Remark~\ref{R:mpedphseest} and the elementary inequality
\[
|e^{\omega_1}-e^{\omega_2}|
\le
\left(|e^{\omega_1}|+|e^{\omega_2}|\right)|\omega_1-\omega_2|,
\qquad \omega_1,\omega_2\in\C,
\]
we have
\[
\left|
e^{-2\ii t(\theta(\xi,k)-\theta(\xi,k_c))}
-e^{-\ii(yz+z^3)}
\right|
\le
Ct|S(\xi,k)|e^{-ct|k-k_c|^3},
\qquad k\in\Sigma_3^\eps.
\]
Therefore,
\[
|\widetilde g_1(k)|
\le
C|k-k_c|e^{-ct|k-k_c|^3}
+Ct|S(\xi,k)|e^{-ct|k-k_c|^3}
+C|d_1(k)-1|e^{-ct|k-k_c|^3}.
\]
Thus \(\widetilde g_1(k)\) obeys the same bound as the function \(g_1(k)\)
defined in~\eqref{E:g1}.  Repeating the same estimates as in the proof of the
corresponding result for \(g_1(k)\), we obtain the desired
\(L^\infty\)- and \(L^1\)-bounds on \(\Sigma_3^\eps\).  The other pieces of
\(\Sigma^\eps\) are treated in exactly the same way, which proves
\eqref{E:estVloc-1}. 
\end{proof}

\medskip
\noindent\textbf{The remaining steps.}
For \((x,t)\in\cP_+\), the outer parametrix is still given by
\(M^{\mathrm{out}}\) defined in~\eqref{E:Mout}.  Therefore, the small-norm RH
problem can be introduced and analyzed exactly as in Subsection~\ref{su:smn}.
Repeating the computations in Subsections~\ref{su:smn} and~\ref{S:pthmain}, we
obtain the same asymptotic formula~\eqref{E:asygs} for \(q(x,t)\), uniformly
for \((x,t)\in\cP_+\) as \(t\to\infty\).

\section{Concluding remarks}

In this paper, we have studied the long-time asymptotics of the focusing NLS
equation with symmetric NZBCs by using the Deift--Zhou nonlinear steepest
descent method.  The main contribution of the paper is twofold.  First, we
complete the Biondini--Mantzavinos long-time asymptotic picture by deriving the
missing boundary-layer formula at the interface between the plane-wave and
modulated elliptic-wave regions.  Second, we show that this transition regime
is governed by a distinguished tritronqu\'ee solution of an inhomogeneous
Painlev\'e-II equation. 

It is also useful to compare our result with the corresponding defocusing
theory.  Wang and Fan~\cite{WF2023} recently showed that, for the scalar defocusing NLS
equation with finite-density initial data, the generic Painlev\'e transition
asymptotics are related to the Hastings--McLeod solution of the homogeneous
Painlev\'e-II equation.  By contrast, the present work shows that the
transition asymptotics for the focusing NLS equation are described by a
tritronqu\'ee solution of an inhomogeneous Painlev\'e-II equation. 

A natural direction for future study is whether the present results can be extended to the
focusing vector NLS system.  To the best of the authors' knowledge, no
long-time asymptotic result is currently available for the focusing vector NLS
equation with NZBCs, although the corresponding IST has been developed in~\cite{BD2015-2,P2023,LSG-2023}.  Our ongoing
study of the focusing Manakov system~\cite{ZGP2026} suggests that the large-time division of
the \((x,t)\)-half-plane is the same as in the scalar focusing NLS case: there
are two plane-wave regions, a central modulated elliptic-wave region, and two
transition regions between them.  It is therefore natural to expect that the
transition asymptotics for the focusing vector problem should exhibit
Painlev\'e-type structures analogous to those found in the scalar problem.
Indeed, this expectation was one of the main motivations for the present work,
whose purpose is to clarify the transition mechanism for the scalar and vector focusing NLS
equations in a systematic way.

We hope that the present work will stimulate further research in this
direction.


\appendix
\section{Painlev\'e-II parametrix}\label{App:PII}
\begin{figure}[htbp]
\centering
\begin{tikzpicture}[
    scale=1.08,
    >=latex,
    line cap=round,
    line join=round,
    every node/.style={font=\small},
    contour/.style={draw=black, line width=1.6pt},
    arrowcontour/.style={
        contour,
        postaction={decorate},
        decoration={
            markings,
            mark=at position 0.55 with {\arrow{Stealth}}
        }
    }
]

\begin{scope}[xshift=-3.9cm]
    \coordinate (A) at (-0.9,0);
    \coordinate (B) at (1.0,0);

    \draw[arrowcontour] (-2.8,1.05) -- (A);
    \draw[arrowcontour] (-2.8,-1.05) -- (A);
    \draw[arrowcontour] (A) -- (B);
    \draw[arrowcontour] (B) -- (1.0,1.35);
    \draw[arrowcontour] (B) -- (1.0,-1.35);

    \node at (-2.25,1.20) {$1$};
    \node at (-2.35,-1.28) {$2$};
    \node at (1.33,0.95) {$3$};
    \node at (1.33,-1.05) {$4$};
    \node at (0.10,-0.25) {$5$};
\end{scope}

\draw[draw=black, line width=0.5pt, dashed] (-0.55,-2.05) -- (-0.55,2.05);

\begin{scope}[xshift=3.15cm, scale=0.72]
    \coordinate (O) at (0,0);

    \draw[arrowcontour] (-3.0,1.75) -- (O);
    \draw[arrowcontour] (-3.0,-1.75) -- (O);
    \draw[arrowcontour] (O) -- (0,2.45);
    \draw[arrowcontour] (O) -- (0,-2.45);

    \node at (-2.0,1.65) {$1$};
    \node at (-2.0,-1.65) {$2$};
    \node at (0.30,1.80) {$3$};
    \node at (0.30,-1.80) {$4$};
\end{scope}

\end{tikzpicture}
\caption{The local contour \(\Sigma^X\).  Left: the case
\((x,t)\in\cP_-\).  Right: the case \((x,t)\in\cP_+\).}
\label{fig:SigmaX}
\end{figure}
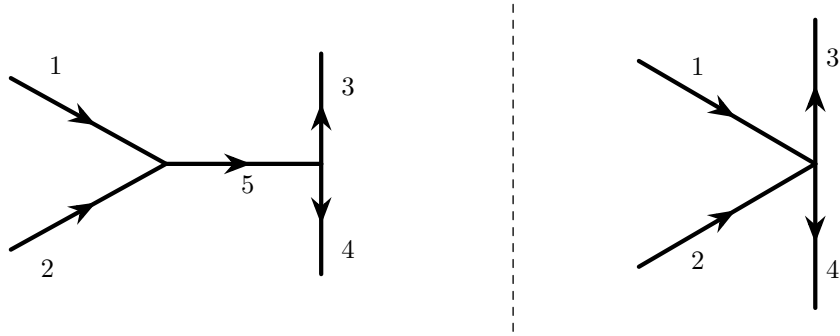

In this appendix we formulate the local Painlev\'e-II model problems used in
Sections~\ref{S:dza} and~\ref{S:dzcpplus}.  The two subregions
\(\cP_-\) and \(\cP_+\) give rise to slightly different local contours.
Nevertheless, both model problems can be reduced to the same inhomogeneous
Painlev\'e-II RH problem.  After recalling the relevant results
from~\cite{Miller}, we apply them to the present local models and derive the
large-\(z\) expansion of the corresponding solutions.  In particular, the
coefficients in these expansions are the same in the two cases.

Define
\[
        \Phi(z;y):=yz+z^3.
\]
Let \(r_0>0\) and \(\nu<0\) be defined by~\eqref{E:r0} and~\eqref{E:nu},
respectively.  We set
\[
\Sigma^X=
\begin{cases}
\displaystyle \bigcup_{j=1}^5\Sigma_j^X,  & (x,t)\in\cP_-,\\[1mm]
\displaystyle \bigcup_{j=1}^4\Sigma_j^X,  & (x,t)\in\cP_+,
\end{cases}
\]
where the contours are displayed in Figure~\ref{fig:SigmaX}.  The local model
RH problem is the following.

\begin{RHP}\label{RHmP-minus}
Find a \(2\times2\) matrix-valued function \(M^X(z;y)\) with the following
properties.

\begin{itemize}
\item \textbf{Analyticity.}
The function \(M^X(\cdot;y)\) is analytic for \(z\in\C\setminus\Sigma^X\).

\item \textbf{Jump condition.}
The boundary values of \(M^X\) on \(\Sigma^X\) exist continuously and satisfy
\[
        M^X_+(z;y)=M^X_-(z;y)V^X(z;y),
        \qquad z\in\Sigma^X.
\]
For \((x,t)\in\cP_-\), the jump matrices are
\(V^X=V_j^X\) on \(\Sigma_j^X\), \(j=1,\ldots,5\).  For
\((x,t)\in\cP_+\), the jump matrices are \(V^X=V_j^X\) on
\(\Sigma_j^X\), \(j=1,\ldots,4\).  They are given by
\be \label{E:Vjexp-minus}
\begin{aligned}
V_1^X(z)
&=
\begin{pmatrix}
1
&
\dfrac{r_0}{1+r_0^2}
\e^{\ii\Phi(z;y)}
z^{2\ii\nu}
\\[1.2ex]
0
&
1
\end{pmatrix},
\qquad
V_2^X(z)
=
\begin{pmatrix}
1
&
0
\\[1.2ex]
\dfrac{r_0}{1+r_0^2}
\e^{-\ii\Phi(z;y)}
z^{-2\ii\nu}
&
1
\end{pmatrix},
\\[2ex]
V_3^X(z)
&=
\begin{pmatrix}
1
&
0
\\[1.2ex]
r_0
\e^{-\ii\Phi(z;y)}
z^{-2\ii\nu}
&
1
\end{pmatrix},
\qquad
V_4^X(z)
=
\begin{pmatrix}
1
&
r_0
\e^{\ii\Phi(z;y)}
z^{2\ii\nu}
\\[1.2ex]
0
&
1
\end{pmatrix},
\\[2ex]
V_5^X(z)
&=
\begin{pmatrix}
1
&
\dfrac{r_0}{1+r_0^2}
\e^{\ii\Phi(z;y)}
z_+^{2\ii\nu}
\\[1.2ex]
\dfrac{r_0}{1+r_0^2}
\e^{-\ii\Phi(z;y)}
z_-^{-2\ii\nu}
&
1+r_0^2
\end{pmatrix}.
\end{aligned}
\ee
Here \(V_5^X\) appears only in the case \((x,t)\in\cP_-\).  We use the notation  \(z_\pm^{\pm 2\ii\nu}\) to  denote the corresponding boundary values on
\(\Sigma_5^X\).

\item \textbf{Normalization.}
As \(z\to\infty\),
\(
        M^X(z;y)=I+\mathcal O(z^{-1}).
\)
\end{itemize}
\end{RHP}

\begin{figure}[htbp]
\centering
\begin{tikzpicture}[
    scale=1.2,
    >=latex,
    line cap=round,
    line join=round,
    every node/.style={font=\small},
    contour/.style={draw=black, line width=1.6pt},
    arrowcontour/.style={
        contour,
        postaction={decorate},
        decoration={
            markings,
            mark=at position 0.55 with {\arrow{Stealth}}
        }
    }
]

    \coordinate (C)  at (-1.3,0);
    \coordinate (D)  at ( 1.25,0);

    \coordinate (S1) at (-0.25, 1.35);
    \coordinate (S2) at (-0.25,-1.35);

    \coordinate (U2) at (-0.45, 0);
    \coordinate (L2) at (-0.45, 0);

    \coordinate (U1) at (-2.507, 1.35);
    \coordinate (L1) at (-2.507,-1.35);

    \fill[gray!25] (U1) -- (S1) -- (D) -- (U2) -- cycle;
    \fill[gray!25] (L1) -- (S2) -- (D) -- (L2) -- cycle;

    \draw[arrowcontour, dashed] (U1) -- (U2);
    \draw[arrowcontour, dashed] (L1) -- (L2);

    \draw[arrowcontour] (-2.9,0) -- (C);

    \draw[arrowcontour] (C) -- (D);

    \draw[arrowcontour] (S1) -- (D);
    \draw[arrowcontour] (S2) -- (D);

    \draw[arrowcontour] (D) -- (1.25, 1.55);
    \draw[arrowcontour] (D) -- (1.25,-1.55);

    \node at (-0.35, 0.55) {\(\mathcal R_1^X\)};
    \node at (-0.35,-0.55) {\(\mathcal R_2^X\)};

\end{tikzpicture}
\caption{Schematic illustration of the gray regions corresponding to \(\{\mathcal R_j^X\}_{j=1}^2\). The black solid lines indicate the jump contour for \(\widehat{M}^X\).}
\label{fig:RX-regions}
\end{figure}
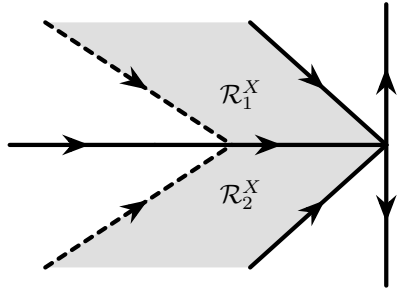

To relate this RH problem to the standard Painlev\'e-II model, we first perform
a simple contour deformation.   We let
\be \label{E:widehatM}
\widehat M^X(z)
:=
M^X(z) z^{\mathrm{i}\nu\sigma_3} \times
\begin{cases}
\displaystyle G(z),  & (x,t)\in\cP_-,\\[1mm]
\displaystyle I,  & (x,t)\in\cP_+,
\end{cases}
\ee
where
\[
G(z)=
\begin{cases}
\begin{pmatrix}
1
&
-\dfrac{r_0}{1+r_0^2}
\e^{\ii \Phi (z,y)}
\\[1.5ex]
0 & 1
\end{pmatrix},
&
z\in \cR_1^X,
\\[4ex]
\begin{pmatrix}
1 & 0
\\[1.5ex]
\dfrac{r_0}{1+r_0^2}
\e^{-\ii \Phi(z,y)}
&
1
\end{pmatrix},
&
z\in \cR_2^X,
\\[4ex]
I,
&
z\notin \cR_1^X\cup \cR_2^X.
\end{cases}
\]
The regions \(\{\mathcal R_j^X\}_{j=1}^2\) and the jump contour for $\widehat M^X(z)$ are shown in Figure~\ref{fig:RX-regions}.

Then \(\widehat M^X(z)\) can be expressed in terms of the solution of  the following RH problem:

\begin{RHP}[Jimbo--Miwa Painlev\'e-II problem]
\label{rhp:PII-generalized}
Let \(y,p,\tau\in\mathbb{C}\) be related by
\(
\tau^2=e^{2\pi p}-1
\).
Seek a \(2\times 2\) matrix-valued function
\(
W(z;y)=W(z;y,p,\tau)
\)
with the following properties.

\begin{itemize}
\item[]\textbf{Analyticity.}
The matrix \(W(z;y)\) is analytic for \(z\) in the five sectors
 $S_0$: $|\arg(z)|<\tfrac{1}{2}\pi$, $S_1$: $\tfrac{1}{2}\pi<\arg(z)<\tfrac{5}{6}\pi$, $S_{-1}$: $-\tfrac{5}{6}\pi<\arg(z)<-\tfrac{1}{2}\pi$, $S_2$: $\tfrac{5}{6}\pi<\arg(z)<\pi$, and $S_{-2}$: $-\pi<\arg(z)<-\tfrac{5}{6}\pi$.
It takes continuous boundary values on the excluded rays and at the origin from
each sector.

\item[] \textbf{Jump conditions.}
The boundary values satisfy
\[
W_+(z;y)
=
W_-(z;y)V^{\mathrm{PII}}(z;y),
\]
where \(V^{\mathrm{PII}}(z;y)\) is the matrix defined on the jump
contour shown in Figure~\ref{fig:zeta-contour-jumps}.

\item[]\textbf{Normalization.}
As \(z\to\infty\), uniformly in all directions,
\(
W (z;y)z^{\mathrm{i}p\sigma_3}\to I
\).
\end{itemize}
\end{RHP}

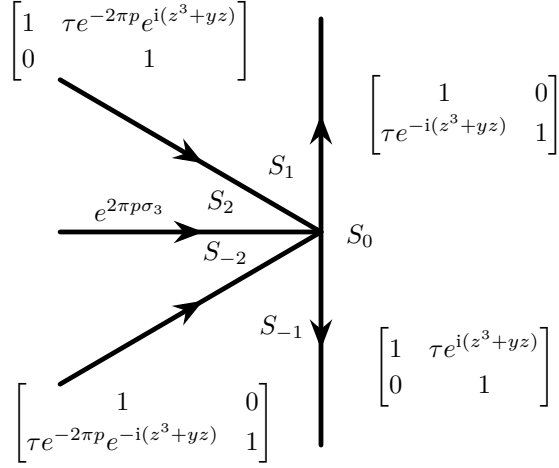
\begin{figure}
\centering
\begin{tikzpicture}[
    scale=1.15,
    >=latex,
    line cap=round,
    line join=round,
    every node/.style={font=\small},
    contour/.style={draw=black, line width=1.8pt},
    arrowcontour/.style={
        contour,
        postaction={decorate},
        decoration={
            markings,
            mark=at position 0.55 with {\arrow{Stealth}}
        }
    }
]

\coordinate (O) at (0,0);

\draw[arrowcontour] (-3.0,1.75) -- (O);
\draw[arrowcontour] (-3.0,-1.75) -- (O);

\draw[arrowcontour] (-3.0,0) -- (O);

\draw[arrowcontour] (O) -- (0,2.45);
\draw[arrowcontour] (O) -- (0,-2.45);

\node at (-0.45,0.75) {$S_1$};
\node at (-1.15,0.35) {$S_2$};
\node at (-1.10,-0.25) {$S_{-2}$};
\node at (-0.45,-1.10) {$S_{-1}$};
\node at (0.45,-0.05) {$S_0$};

\node at (-2.20,0.25) {$e^{2\pi p\sigma_3}$};

\node at (-2.20,2.20) {$
\begin{bmatrix}
1 & \tau e^{-2\pi p}e^{\mathrm{i}(z^3+yz)}
\\[0.4ex]
0 & 1
\end{bmatrix}
$};

\node at (1.65,1.35) {$
\begin{bmatrix}
1 & 0
\\[0.4ex]
\tau e^{-\mathrm{i}(z^3+yz)} & 1
\end{bmatrix}
$};

\node at (1.65,-1.55) {$
\begin{bmatrix}
1 & \tau e^{\mathrm{i}(z^3+yz)}
\\[0.4ex]
0 & 1
\end{bmatrix}
$};

\node at (-2.05,-2.20) {$
\begin{bmatrix}
1 & 0
\\[0.4ex]
\tau e^{-2\pi p}e^{-\mathrm{i}(z^3+yz)} & 1
\end{bmatrix}
$};


\end{tikzpicture}
\caption{The  jump matrices and jump  contour for $W(z;y)$.}
\label{fig:zeta-contour-jumps}
\end{figure}

The RH problem above was studied systematically in~\cite{Miller}.  In
particular, for the parameter range relevant here, it has a unique solution.
Moreover, the product \(W(z;y)z^{\ii p\sigma_3}\) admits a complete asymptotic
expansion
\be \label{E:Wasy}
W(z;y)z^{\mathrm{i}p\sigma_3}
\sim
I+\sum_{j=1}^{\infty} W_j(y) z^{-j},
\qquad
z\to\infty,
\ee
uniformly in all directions of the \(z\)-plane.  In addition, define
\[
\mathcal V(y)
:=
\lim_{z\to\infty}
z \,W_{21}(z;y)z^{\mathrm{i}p}
=
( W_1(y))_{21}.
\]
This function can be described in terms of a special solution of an
inhomogeneous Painlev\'e-II equation.  More precisely, there exists a unique
tritronqu\'ee solution \(\Q(y)\) of
\[
\frac{d^2 \Q}{dy^2}
+
\frac{2}{3}y \Q
-
2 \Q^3
-
\frac{2}{3}\mathrm{i}p
-
\frac{1}{3}
=0,
\]
which satisfies the asymptotic condition
\[
\Q(y)
=
\mathrm{i}\left(-\frac{y}{3}\right)^{1/2}
-
\left(\frac14+\frac{\mathrm{i}p}{2}\right)\frac{1}{y}
+
\mathcal{O}\left(|y|^{-5/2}\right),
\qquad
y\to-\infty.
\]
This solution is analytic for real \(y\) and has oscillatory/algebraic
asymptotics as \(y\to+\infty\).  

According to~\cite{Miller}, the logarithmic derivative of
\(\mathcal V(y)\) satisfies the above inhomogeneous Painlev\'e-II equation.
More precisely,
\[
\Q(y)=\frac{\mathcal V'(y)}{\mathcal V(y)}.
\]
For \(y\to-\infty\), one has
\[
\begin{aligned}
\mathcal V(y)
&=
\frac{\tau p\Gamma(\mathrm{i}p)}{2\sqrt{\pi}}\,
e^{-\frac{3\pi\mathrm{i}}{4}}
e^{-\frac{\pi p}{2}}
2^{-\mathrm{i}p}
e^{-2\mathrm{i}\left(-\frac{y}{3}\right)^{3/2}}
(-3 y)^{-\frac14-\frac{\mathrm{i}p}{2}}
\left[1+\mathcal{O}\left(|y|^{-5/4}\right)\right],
\\
&\hspace{8cm} y\to-\infty.
\end{aligned}
\]
Thus,  \(\mathcal V(y)\) can be written in the following form:
\be \label{E:defmathcalV}
\mathcal V(y)=
\begin{cases}
\displaystyle
C_{p,\tau}
e^{-\frac{2}{9} \sqrt{3}\ii \left(- y\right)^{3/2}}
(- 3 y)^{-\frac14-\frac{\ii p}{2}}
\\[1mm]
\displaystyle\qquad\qquad\qquad
{}\times
\exp\left\{
\int_{-\infty}^{y}
\left[
\Q(s)-\ii\left(-\frac{s}{3}\right)^{1/2}
+
\left(\frac14+\frac{\ii p}{2}\right)\frac{1}{s}
\right]\mathrm ds
\right\},
& y<0,
\\[5mm]
\displaystyle
\mathcal V(-1)\exp\left\{
\int_{-1}^{y}\Q(s)\,\mathrm ds
\right\},
& y\ge 0,
\end{cases}
\ee
where
\be \label{E:MX1esp}
C_{p,\tau}
=
\frac{\tau p\Gamma(\mathrm{i}p)}{2\sqrt{\pi}}\,
e^{-\frac{3\pi\mathrm{i}}{4}}
e^{-\frac{\pi p}{2}}
2^{-\mathrm{i}p}.
\ee
After the specialization \(p=-\nu\) and $\tau=r_0$, we may write
\[
C_{p,\tau}=\alpha_0e^{\mathrm{i}\phi_0},
\]
with
\[
\alpha_0=|C_{p,\tau}|=\sqrt{-\frac{\nu}{2}},
\qquad
\phi_0
=
-\frac{3}{4}\pi
+\nu\log 2
+\arg\Gamma(-\mathrm{i}\nu).
\]

We now apply these facts to the local model \(M^X\).    After the transformation~\eqref{E:widehatM},
both cases are reduced to the same Painlev\'e-II model with
\(p=-\nu\) and \(\tau=r_0\).  Moreover, the factor \(G(z)\) in
\eqref{E:widehatM} tends to the identity matrix exponentially fast as
\(z\to\infty\) along the relevant sectors.  Therefore it does not contribute
to the algebraic coefficients in the large-\(z\) expansion. 
Thus, in both cases, the local solutions have the same coefficients in their
large-\(z\) expansions.

\begin{lemma}
 \(M^X(z;y)\) admits the following expansion as \(z\to\infty\):
\be \label{E:asyMX}
M^X(z;y)
=
I+\frac{M_1^X(y)}{z}
+
\mathcal O\left(\frac{1}{z^2}\right),
\ee
uniformly for \(\arg z\in[0,2\pi]\).  Moreover,
\be \label{E:defMX1dd}
M_1^X(y)
=
\begin{pmatrix}
\widehat{ \mathcal V}(y) & -\mathcal V^*(y)
\\[0.6ex]
\mathcal V(y) & \widehat{ \mathcal V}^*(y)
\end{pmatrix},
\ee
where \(\mathcal V(y)\) is defined by~\eqref{E:defmathcalV} with the parameters \(p=-\nu\) and \(\tau=r_0\), and
\[
\widehat {\mathcal V}(y)=
\begin{cases}
\displaystyle
\ii\nu\left(-\frac{y}{3}\right)^{1/2}
+
\ii\int_{-\infty}^{y}
\left[
|\mathcal V (s)|^2
+
\frac{\nu}{2\sqrt{3}}(-s)^{-1/2}
\right]\,\mathrm ds,
& y<0,
\\[5mm]
\displaystyle
\widehat{\mathcal V}(-1)
+
\ii\int_{-1}^{y}|\mathcal V (s)|^2\,\mathrm ds,
& y\ge0.
\end{cases}
\]

\end{lemma}

\begin{proof}
A direct verification of the jump matrices shows that
\be \label{E:ygds}
        \widehat M^X(z;y)=W(z;y,p,\tau),
        \qquad p=-\nu,
        \qquad \tau=r_0.
\ee
Here \(y\in\R\), \(-\nu>0\), and \(r_0>0\).  With this choice of parameters,
the jump matrices satisfy the Schwarz symmetry
\[
\left[V^{\mathrm{PII}}(z^*;y)^*\right]^{-1}
=
\sigma_2 V^{\mathrm{PII}}(z;y)\sigma_2.
\]
By uniqueness of the RH problem, it follows that
\[
        W(z^*;y)^*=\sigma_2W(z;y)\sigma_2.
\]
Consequently, the coefficient \(W_1(y)\) in~\eqref{E:Wasy} satisfies
\[
        W_1(y)^*=\sigma_2W_1(y)\sigma_2,
\]
and hence has the form
\[
W_1(y)
=
\begin{pmatrix}
W_{1,11}(y) & -W_{1,21}^*(y)
\\[0.6ex]
W_{1,21}(y) & W_{1,11}^*(y)
\end{pmatrix}.
\]
Since \(G(z)\) contributes only exponentially small terms at infinity, it
follows from~\eqref{E:widehatM}, \eqref{E:ygds} and
\eqref{E:Wasy} that \(M^X\) admits the expansion~\eqref{E:asyMX} and that
\[
        M_1^X(y)=W_1(y).
\]
By definition,
\[
        W_{1,21}(y)=\mathcal V(y).
\]
It remains to identify \(W_{1,11}(y)\).

From~\cite[Eq. 2.4]{Miller}, one has
\[
\frac{d W_1(y)}{dy}
+
\frac{\ii}{2}[ W_2(y),\sigma_3]
-
\frac{\mathrm{i}}{2}
[ W_1(y),\sigma_3]  W_1(y)
=
0.
\]
Taking the \((1,1)\)-entry yields
\be \label{E:A11}
\frac{dW_{1,11}}{dy}
=
-\mathrm{i}W_{1,12}W_{1,21}
=
\mathrm{i}|W_{1,21}|^2
=
\mathrm{i}|\mathcal V(y)|^2.
\ee
On the other hand, Section~2.2.1 of~\cite{Miller} gives
\[
W_{1,11}(y)
=
-\mathrm{i}p\left(-\frac{y}{3}\right)^{1/2}
+
\mathcal{O}\left(|y|^{-1}\right),
\qquad
y\to-\infty.
\]
Furthermore,
\[
|\mathcal V(y)|^2
=
\frac{p}{2\sqrt{3}}(-y)^{-1/2} + \mathcal{O}(|y|^{-7/4}), 
\qquad
y\to-\infty.
\]
Therefore \(W_{1,11}(y)\) is determined by integration from \(-\infty\):
\[
W_{1,11}(y)
=
-\mathrm{i}p\left(-\frac{y}{3}\right)^{1/2}
+
\mathrm{i}
\int_{-\infty}^{y}
\left[
|\mathcal V(s)|^2
-
\frac{p}{2\sqrt{3}}(-s)^{-1/2}
\right]\,ds, \qquad y<0.
\]
Since the value of \(W_{1,11} (-1)\) can be computed from the formula above,
integrating both sides of~\eqref{E:A11} from \(-1\) to \(y\) gives the value of
\(W_{1,11}(y)\) for any \(y\ge 0\).
Thus, by setting \(\widehat{\mathcal V}=W_{1,11}\),  the asserted form of \(M_1^X(y)\) follows.
\end{proof}

\section*{Acknowledgments}
This work is supported by National Natural Science Foundation of China (Grant Nos. 12471234, 12471240, 12571268) and Science Foundation of Henan Academy of Sciences (Grant No. 20252319002).\\

\noindent{\bf Data Availability Statement} Data sharing not applicable to this article as no datasets were generated or analyzed during the current study.

\subsection*{Declarations}

{\bf Conflict of interest statement} We declare that there is no conflict of interests.



\begin{thebibliography}{99}
\addcontentsline{toc}{section}{References}


\bibitem{APT2004}
M. J. Ablowitz, B. Prinari, and A. D. Trubatch,
\emph{Discrete and Continuous Nonlinear Schr\"odinger Systems},
London Math. Soc. Lecture Note Ser., vol. 302, Cambridge University Press,
Cambridge, 2004.

\bibitem{AH1981}
M. J. Ablowitz and H. Segur,
\emph{Solitons and the Inverse Scattering Transform},
SIAM Stud. Appl. Math., vol. 4, Society for Industrial and Applied
Mathematics (SIAM), Philadelphia, PA, 1981.

\bibitem{EG2003}
E. Infeld and G. Rowlands,
\emph{Nonlinear Waves, Solitons and Chaos}, 2nd ed.,
Cambridge University Press, Cambridge, 2000.

\bibitem{CP1999}
C. Sulem and P.-L. Sulem,
\emph{The Nonlinear Schr\"odinger Equation: Self-Focusing and Wave Collapse},
Appl. Math. Sci., vol. 139, Springer-Verlag, New York, 1999.

\bibitem{Miller}
P. D. Miller, On the increasing tritronqu\'ee solutions of the Painlev\'e-II
equation, \emph{SIGMA Symmetry Integrability Geom. Methods Appl.}
\textbf{14} (2018), Paper No. 125, 38 pp.

\bibitem{BM2017}
G. Biondini and D. Mantzavinos, Long-time asymptotics for the focusing
nonlinear Schr\"odinger equation with nonzero boundary conditions at infinity
and asymptotic stage of modulational instability,
\emph{Comm. Pure Appl. Math.} \textbf{70} (2017), no. 12, 2300--2365.

\bibitem{BLM2021}
G. Biondini, S. Li, and D. Mantzavinos, Long-time asymptotics for the focusing
nonlinear Schr\"odinger equation with nonzero boundary conditions in the
presence of a discrete spectrum, \emph{Comm. Math. Phys.} \textbf{382}
(2021), no. 3, 1495--1577.



\bibitem{B1967-1}
T. B. Benjamin, Instability of periodic wavetrains in nonlinear dispersive
systems, \emph{Proc. Roy. Soc. London Ser. A} \textbf{299} (1967),
no. 1456, 59--75.

\bibitem{B1967-2}
T. B. Benjamin and J. E. Feir, The disintegration of wave trains on deep water.
Part 1. Theory, \emph{J. Fluid Mech.} \textbf{27} (1967), no. 3,
417--430.


\bibitem{AFM2019}
M. A. Alejo, L. Fanelli, and C. Mu\~noz, The Akhmediev breather is unstable,
\emph{S\~ao Paulo J. Math. Sci.} \textbf{13} (2019), no. 2, 391--401,
doi:10.1007/s40863-019-00145-4.


\bibitem{AFM2020}
M. A. Alejo, L. Fanelli, and C. Mu\~noz, Review on the stability of the
Peregrine and related breathers, \emph{Front. Phys.} \textbf{8} (2020),
591995, doi:10.3389/fphy.2020.591995.

\bibitem{ZO2009}  V.E. Zakharov and  L.A. Ostrovsky,  Modulation instability: the beginning, Phys. D. \textbf{238} (2009), 540–548.

\bibitem{BF2015}
G. Biondini and E. Fagerstrom, The integrable nature of modulational instability, \textit{SIAM J. Appl. Math.} \textbf{75} (2015), 136--163.

\bibitem{BLMT2018}
G. Biondini, S. Li, D. Mantzavinos, and S. Trillo, Universal behavior of modulationally unstable media, \textit{SIAM Rev.} \textbf{60} (2018), 888--908.

\bibitem{zs1971}
V. E. Zakharov and A. B. Shabat, Exact theory of two-dimensional
self-focusing and one-dimensional self-modulation of waves in nonlinear
media, \emph{Soviet Phys. JETP} \textbf{34} (1972), no. 1, 62--69.


\bibitem{pz93}
P. Deift and X. Zhou, A steepest descent method for oscillatory
Riemann--Hilbert problems. Asymptotics for the mKdV equation,
\emph{Ann. of Math. (2)} \textbf{137} (1993), no. 2, 295--368.

\bibitem{BK2014}
G. Biondini and G. Kova\v{c}i\v{c}, Inverse scattering transform for the
focusing nonlinear Schr\"odinger equation with nonzero boundary conditions,
\emph{J. Math. Phys.} \textbf{55} (2014), no. 3, 031506, 24 pp.


\bibitem{m2017}
C. Mu\~noz, Instability in nonlinear Schr\"odinger breathers,
\emph{Proyecciones} \textbf{36} (2017), no. 4, 653--683.

\bibitem{Monvel5}
A. Boutet de Monvel, J. Lenells, and D. Shepelsky, The focusing NLS equation
with step-like oscillating background: asymptotics in a transition zone,
\emph{J. Differential Equations} \textbf{429} (2025), 747--801.

\bibitem{LLY2026-JDE}
G. Li, L. Liu, and Y. Yang, Complex Painlev\'e type transient asymptotics of
the focusing NLS equation: step-like oscillating background,
\emph{J. Differential Equations} \textbf{465} (2026), Paper No. 114216, 35 pp.

\bibitem{DLM2020}
D. Bilman, L. Ling, and P. D. Miller, Extreme superposition: rogue waves of
infinite order and the Painlev\'e-III hierarchy, \emph{Duke Math. J.}
\textbf{169} (2020), no. 4, 671--760.

\bibitem{DP2024}
D. Bilman and P. D. Miller,
\emph{Extreme Superposition: High-Order Fundamental Rogue Waves in the
Far-Field Regime}, \emph{Mem. Amer. Math. Soc.} \textbf{300} (2024),
no. 1505, v+90 pp.

\bibitem{boutet2011}
A. Boutet de Monvel, V. P. Kotlyarov, and D. Shepelsky, Focusing NLS equation:
long-time dynamics of step-like initial data, \emph{Int. Math. Res. Not. IMRN}
\textbf{2011} (2011), no. 7, 1613--1653.

\bibitem{CLPa2020}
C. Charlier and J. Lenells, Airy and Painlev\'e asymptotics for the mKdV
equation, \emph{J. Lond. Math. Soc. (2)} \textbf{101} (2020), no. 1,
194--225.

\bibitem{BJM2018}
M. Borghese, R. Jenkins, and K. D. T.-R. McLaughlin, Long time asymptotic
behavior of the focusing nonlinear Schr\"odinger equation,
\emph{Ann. Inst. H. Poincar\'e C Anal. Non Lin\'eaire} \textbf{35}
(2018), no. 4, 887--920.

\bibitem{PZ2002}
P. Deift and X. Zhou, Long-time asymptotics for solutions of the NLS equation
with initial data in a weighted Sobolev space, \emph{Comm. Pure Appl. Math.}
\textbf{56} (2003), no. 8, 1029--1077.

\bibitem{BIS2010}
A. Boutet de Monvel, A. R. Its, and D. Shepelsky, Painlev\'e-type asymptotics
for the Camassa--Holm equation, \emph{SIAM J. Math. Anal.} \textbf{42}
(2010), no. 4, 1854--1873.

\bibitem{WF2023}
Z. Wang and E. Fan, The defocusing nonlinear Schr\"odinger equation with a
nonzero background: Painlev\'e asymptotics in two transition regions,
\emph{Comm. Math. Phys.} \textbf{402} (2023), no. 3, 2879--2930.

\bibitem{CuJe2016}
S. Cuccagna and R. Jenkins, On the asymptotic stability of $N$-soliton
solutions of the defocusing nonlinear Schr\"odinger equation,
\emph{Comm. Math. Phys.} \textbf{343} (2016), no. 3, 921--969.

\bibitem{CL2024main}
C. Charlier and J. Lenells, On Boussinesq's equation for water waves,
preprint, arXiv:2204.02365.

\bibitem{WZZ2026}
D.-S. Wang, D. Zhu, and X. Zhu, Genus two KdV soliton gases and their
long-time asymptotics, \emph{Forum Math. Sigma} \textbf{14} (2026),
Paper No. e57, 47 pp.

\bibitem{HWZ2026}
L. Huang, D.-S. Wang, and X. Zhu, Long-time asymptotics of the Tzitz\'eica
equation on the line, \emph{Math. Ann.} \textbf{395} (2026), no. 1,
Paper No. 18.

\bibitem{YTL2025}
J.-J. Yang, S.-F. Tian, and Z.-Q. Li, The modified Camassa--Holm equation
with nonzero background: soliton resolution conjecture and asymptotic
stability of $N$-soliton solutions, \emph{Adv. Math.} \textbf{484}
(2026), Paper No. 110552.

\bibitem{FLYZ2026}
E. Fan, G. Li, Y. Yang, and L. Zhang, Painlev\'e XXXIV asymptotics for the
defocusing nonlinear Schr\"odinger equation with a finite-genus
algebro-geometric background, \emph{Math. Ann.} \textbf{394} (2026),
Paper No. 44, 58 pp.

\bibitem{CLW2023}
C. Charlier, J. Lenells, and D.-S. Wang, The ``good'' Boussinesq equation:
long-time asymptotics, \emph{Anal. PDE} \textbf{16} (2023), no. 6,
1351--1388.

\bibitem{CJ2024}
C. Charlier and J. Lenells, The soliton resolution conjecture for the
Boussinesq equation, \emph{J. Math. Pures Appl. (9)} \textbf{191}
(2024), Paper No. 103621.

\bibitem{DKMVZ1999}
P. Deift, T. Kriecherbauer, K. T.-R. McLaughlin, S. Venakides, and X. Zhou,
Uniform asymptotics for polynomials orthogonal with respect to varying
exponential weights and applications to universality questions in random matrix
theory, \emph{Comm. Pure Appl. Math.} \textbf{52} (1999), no. 11,
1335--1425.

\bibitem{Deift1999}
P. Deift, \emph{Orthogonal Polynomials and Random Matrices: A Riemann--Hilbert
Approach}, Courant Lecture Notes in Mathematics, vol. 3, New York University,
Courant Institute of Mathematical Sciences, New York; American Mathematical
Society, Providence, RI, 1999.


\bibitem{Lenells2018}
J. Lenells, Matrix Riemann--Hilbert problems with jumps across Carleson
contours, \emph{Monatsh. Math.} \textbf{186} (2018), no. 1, 111--152.

\bibitem{Lenells2017}
J. Lenells, The nonlinear steepest descent method for Riemann--Hilbert
problems of low regularity, \emph{Indiana Univ. Math. J.} \textbf{66}
(2017), no. 4, 1287--1332.



\bibitem{BD2015-2}
D. Kraus, G. Biondini, and G. Kova\v{c}i\v{c}, The focusing Manakov system
with nonzero boundary conditions, \emph{Nonlinearity} \textbf{28} (2015),
no. 9, 3101--3151.

\bibitem{P2023}
B. Prinari, Inverse scattering transform for nonlinear Schr\"odinger systems
on a nontrivial background: a survey of classical results, new developments
and future directions, \emph{J. Nonlinear Math. Phys.} \textbf{30} (2023),
no. 2, 317--383.

\bibitem{LSG-2023}
H. Liu, J. Shen, and X. G. Geng, Inverse scattering transformation for the
$N$-component focusing nonlinear Schr\"odinger equation with nonzero boundary
conditions, \emph{Lett. Math. Phys.} \textbf{113} (2023), no. 1,
Paper No. 23, 47 pp.

\bibitem{ZGP2026}
X. G. Geng and H. B. Zhang, Long-time asymptotics for the focusing Manakov
system with nonzero boundary conditions, in preparation.


\end{thebibliography}
\end{document}